\documentclass[12pt,english]{article}
\usepackage[T1]{fontenc}
\usepackage[latin9]{inputenc}
\usepackage[a4paper]{geometry}
\usepackage{babel}
\usepackage{verbatim}
\usepackage{float}
\usepackage{booktabs}
\usepackage{mathtools}
\usepackage{enumitem}
\usepackage{amsmath}
\usepackage{amsthm}
\usepackage{amssymb}
\usepackage{graphicx}
\graphicspath{{revision two/}}

\usepackage{setspace}
\usepackage[authoryear]{natbib}
\usepackage[unicode=true]
 {hyperref}

\makeatletter

\providecommand{\tabularnewline}{\\}

\theoremstyle{definition}
\newtheorem{example}{\protect\examplename}[section]
\theoremstyle{definition}
\newtheorem{defn}{\protect\definitionname}[section]
\theoremstyle{plain}
\newtheorem{assumption}{\protect\assumptionname}
\theoremstyle{plain}
\newtheorem{thm}{\protect\theoremname}[section]
\theoremstyle{plain}
\newtheorem{prop}{\protect\propositionname}[section]
\theoremstyle{plain}
\newtheorem{cor}{\protect\corollaryname}[section]
\theoremstyle{plain}
\newtheorem{lem}{\protect\lemmaname}[section]
\theoremstyle{remark}
\newtheorem{claim}{\protect\claimname}[section]

\@ifundefined{date}{}{\date{}}
\newcommand{\gmistake}{\widehat{g}_{\alpha}}
\usepackage{epsf}

\usepackage{babel}

\usepackage{mathtools}
\numberwithin{assumption}{section}
\numberwithin{figure}{section}
\numberwithin{table}{section}
\@ifundefined{showcaptionsetup}{}{%
 \PassOptionsToPackage{caption=false}{subfig}}
\usepackage{subfig}
\makeatother

\providecommand{\assumptionname}{Assumption}
\providecommand{\claimname}{Claim}
\providecommand{\corollaryname}{Corollary}
\providecommand{\definitionname}{Definition}
\providecommand{\examplename}{Example}
\providecommand{\lemmaname}{Lemma}
\providecommand{\propositionname}{Proposition}
\providecommand{\theoremname}{Theorem}
\newcommand{\E}{\ensuremath{\mathbb{E}}}
\renewcommand{\Pr}{\operatorname{Pr}}

\begin{document}
\title{Empirical Welfare Maximization with Constraints\thanks{{\footnotesize{}I am grateful to Alberto Abadie, Isaiah Andrews, Amy
Finkelstein and Anna Mikusheva for their guidance and support. I thank
Ben Deaner, Chishio Furukawa, Jasmin Fliegner, Jon Gruber, Yunan Ji,
Joseph Kibe, Sylvia Klosin, Kevin Kainan Li, \'{A}ureo de Paula, James Poterba, David Preinerstorfer,  Emma Rackstraw,
Jonathan Roth, Dan Sichel, and Cory Smith for helpful discussions.
The paper also benefited from the comments of conference and seminar audiences. 
 I acknowledge generous support from the Jerry
A. Hausman Graduate Dissertation Fellowship and  from the Economic and Social Research Council (new investigator grant UKRI607) . %
}}}
\author{Liyang Sun\thanks{UCL and CEMFI. Email: \protect\href{mailto:liyang.sun@ucl.ac.uk}{liyang.sun@ucl.ac.uk}}}
\date{November, 2025\vspace{-0.3in}}
\maketitle
\begin{abstract}
Empirical Welfare Maximization (EWM) is a framework that can be used to select welfare program eligibility policies based on data.  This paper extends EWM by allowing for uncertainty in estimating the budget needed to implement the selected policy, in addition to its welfare. Due to the additional estimation error, I show there exist no rules that achieve the highest welfare possible while satisfying a budget constraint uniformly over a wide range of DGPs. This differs from the setting without a budget constraint where uniformity is achievable.  I propose an alternative trade-off rule and illustrate it with Medicaid expansion, a setting with imperfect take-up and varying program costs.
\end{abstract}
\begin{onehalfspace}
\end{onehalfspace}


\noindent \newpage{}

\section{Introduction}

When a welfare program induces varying benefits across individuals,
and when resources are scarce, policymakers naturally want to prioritize
eligibility to individuals who will benefit the most. Based on experimental
data, cost-benefit analysis can inform policymakers on which subpopulations
to prioritize, but these subpopulations might not align with any available
eligibility policy such as an income threshold. \citet{kitagawa_who_2018}
propose a statistical rule, Empirical Welfare Maximization (EWM),
that can directly select an eligibility policy from a set of available
policies based on the experimental data. For example, if available
policies take the form of income thresholds, EWM considers the problem
of maximizing the expected benefits in the population
\[
\max_{t\leq\overline{t}}\ \E[\text{benefit}\cdot\mathbf{1}\{\text{income}\leq t\}]
\]
 and approximates the optimal threshold based on benefits estimated
from experimental data. Recent work has demonstrated that the EWM approach performs well across a broad range of data distributions. As the sample size grows, the average of benefits obtained under the eligibility policy selected by EWM converges to the highest attainable level, a property I refer to as as\emph{ uniform asymptotic welfare efficiency}.

I follow this line of work in focusing on settings where the eligibility policy for welfare programs must be determined ex ante and cannot be easily adjusted during implementation, such as Food Stamp and Medicaid in the United States.\footnote{There are also settings in which eligibility is implemented sequentially until the budget is exhausted, such as anti-poverty programs in Chile as analyzed by \cite{carneiro2019tackling}. My analysis does not extend to those cases, as it does not account for the possibility that the benefits and costs may vary depending on the order in which individuals enroll.} In practice policymakers often face budget constraints, but only have
imperfect information about whether a given eligibility policy
satisfies the budget constraint. First, there may be imperfect take-up:
eligible individuals might not participate in the welfare program,
resulting in zero cost to the government, e.g. \citet{finkelstein_take-up_2019}.
Second, costs incurred by eligible individuals who participate in
the welfare program may vary considerably, largely driven by individuals'
different needs but also many other factors, e.g. \citet{finkelstein_adjusting_2017}.
Both considerations are hard to predict ex ante, implying that the
potential cost of providing eligibility to any given individual is
unknown at the time of designing the eligibility policy. Unobservability
of the potential cost requires estimation based on experimental data,
contributing to uncertainty in the budget estimate of a given eligibility
policy.%

For this empirically relevant setting where the budget needed to implement an eligibility policy involves an unknown cost, this paper introduces a new property of statistical rules, namely \emph{asymptotic feasibility}. Policies are \emph{feasible} if they  satisfy the budget constraint in the target population; otherwise, they are referred to as \emph{infeasible}.  A statistical rule is \emph{asymptotically feasible} if given a large enough experimental sample, the statistical rule is very likely to select feasible eligibility policies. While budget overruns are generally tolerated in countercyclical programs like unemployment insurance and Medicaid in the United States, in other settings, such as subsidized health and education programs in developing countries, a potential budget overrun due to an underestimated budget can be highly undesirable for policymakers, as securing additional funding may be difficult.\footnote{For example, international and national funding for malaria control has fallen short of what is estimated to be needed in recent years. Motivated by the fact that Kenya can afford to distribute bed net subsidies to only 50\% of its target population in 2007, \cite{bhattacharya_inferring_2012} analyzed the constrained optimal allocation to increase take-up.} In this context, ensuring \emph{asymptotic feasibility} is particularly desirable, especially for policymakers who are highly risk-averse to budget overruns. %

This paper answers three questions in the current setting with unknown cost: is there any statistical rule that achieves uniformly good performance for a wide range of data distributions in terms of both welfare efficiency and feasibility, whether the obvious extension of the existing EWM statistical rule remains attractive, and what are some alternative statistical rules.

Firstly as a novel theoretical contribution, I quantify a class of reasonable data distributions that is particularly challenging for statistical rules.  Specifically, no statistical rule can be uniformly welfare-efficient and feasible simultaneously over this class of data distributions. A notable example within this class is the expansion of welfare programs, such as tax credits to incentivize labor force participation, which can ``pay for themselves,'' as demonstrated in \cite{hendren_unified_2020}, because they have zero net cost to the government. This impossibility result provides theoretical characterization of constrained settings that need to be ruled out for any statistical rule to achieve uniformity. %

Second, I show the direct extension of the existing EWM statistical rule is not appealing in the setting with unknown cost. The reason is that this \emph{sample-analog rule} ignores the estimation error in the estimated budget needed to implement a given policy, which has non-negligible consequences even when the sample size is large. For data distributions mentioned earlier where the budget constraint is exactly binding, the welfare loss does not vanish with sample size. %
The probability of selecting infeasible policies also does not vanish with sample size. Intuitively, this latter issue can be mitigated by using a slightly downward-biased version of the budget constraint, effectively making the budget estimate more conservative.\footnote{Similar intuition arises in optimal prediction under asymmetric loss \citep{christoffersen1997optimal}, where positive prediction errors are more costly than negative ones, making a downward-biased predictor optimal.} I formalize this idea by setting the degree of conservativeness proportional to the standard error of the budget estimate, effectively selecting only policies with an upper confidence bound below the budget constraint. This modification to the \emph{sample-analog rule} ensures that feasible policies are selected with high probability. %

So far, the discussion assumes that even minor budget violations are unacceptable. However, in practice, exceeding the budget constraint may be desirable if the welfare gains outweigh the borrowing costs.\footnote{What happens when a welfare program goes over budget is highly context-specific, and there are various ways to model it. In this paper, I focus specifically on the scenario where the government borrows money to cover the excess cost while other models-such as those involving rationing as in \cite{kitagawa_who_2018}  and  poorer services are certainly possible. } To account for this situation, I introduce a new objective function that  maximizes welfare gains but imposes some penalty once spending exceeds the budget, thereby giving the new objective function the intuitive interpretation of a \emph{trade-off} function.  I propose the \emph{trade-off rule}, which optimizes the sample-analog version of the trade-off function. I show the trade-off rule is uniformly asymptotically welfare efficient.

To illustrate the trade-off rule, I apply it to data from the Oregon Health Insurance Experiment
(OHIE)  to select 
a more flexible Medicaid expansion policy than the current one. Medicaid is a government-sponsored
health insurance program intended for the low-income population in
the United States. The current Medicaid expansion policy determines eligibility solely based on household income.  I examine whether health can be improved by allowing the income threshold to vary by the number of children in the household, setting the budget constraint equal to the cost of the current policy. Imposing a reasonable penalty for exceeding the budget in this context, that any overrun needs to be repaid in full, the trade-off rule selects an eligibility policy that expands eligibility
for many households above the current threshold, especially those with
children. This occurs because, based on the OHIE data,
 the additional
health benefit from extending eligibility to these households outweigh the  penalty of having to repay the overrun in full. 

The rest of the paper proceeds as follows. Section 1.1 discusses related
work in more detail. Section 2 presents  theoretical results. Sections 3 and 4 discuss properties of two statistical rules, illustrated by an empirical example of designing a more flexible Medicaid expansion policy for the low-income population in Oregon. 
Section 5 conducts a simulation study and Section 6 concludes.
Proofs, supporting lemmas, additional results and computational details can be found in the Appendix.

\subsection{Literature review}

This paper is related to the traditional literature on cost-benefit
analysis, e.g. \citet{dhailiwal_comparative_2013}, and to the recent
literature on EWM, e.g. \citet{kitagawa_who_2018}, \citet{rai_statistical_2019},
\citet{athey_policy_2021} and \citet{mbakop_model_2021}. More broadly,
this paper contributes to a growing literature on statistical rules
in econometrics, including \citet{manski_statistical_2004}, \citet{dehejia_program_2005},
\citet{hirano_asymptotics_2009}, \citet{stoye_minimax_2009}, \citet{chamberlain_bayesian_2011},
\citet{bhattacharya_inferring_2012}, \citet{demirer_semi-parametric_2019},
\citet{yata_optimal_2021} and \citet{kitagawa_treatment_2022}, among
others.

The traditional cost-benefit analysis compares the cost and benefit
of a given welfare program. The effect of program eligibility is first
estimated based on a randomized control trial (RCT), and then converted
to a monetary benefit for calculating the cost-benefit ratio. For
example, \citet{gelber_effects_2016} and \citet{heller_thinking_2017}
compare the efficiency of various crime prevention programs based
on their cost-benefit ratios. However, the cost-benefit ratio is only
informative for whether this welfare program should be implemented
with the fixed eligibility policy as implemented in the RCT. 

The literature on statistical rules in econometrics has also developed
a definition for optimality of statistical rules. \citet{manski_statistical_2004}
considers the regret, defined to be loss in expected welfare
achieved by the statistical rule relative to the welfare achieved
by the theoretically optimal eligibility policy. In the absence
of any constraint, under the theoretically optimal eligibility policy,
anyone with positive benefit from the welfare program would be assigned
with eligibility. The minimax regret rule minimizes the upper bound on the regret
that results from not knowing the data distribution. %
\citet{stoye_minimax_2009}
shows that with continuous covariates and no functional form restrictions
on the set of policies, minimax regret does not converge to zero with
the sample size because the theoretically optimal policy can be
too difficult to approximate by a statistical rule. \citet{kitagawa_who_2018}
avoid this issue by imposing functional form restrictions. They propose
the EWM rule, which starts with functional form restrictions on the
class of available policies, and then selects the policy with the
highest estimated benefit (empirical welfare) based on an RCT sample.
They prove the optimality of EWM in the sense that its regret converges
to zero at the minimax rate. Importantly, the regret is defined to
be loss in expected welfare relative to the maximum achievable welfare
in the constrained class, which avoids the negative results of \citet{stoye_minimax_2009}.
\citet{athey_policy_2021} propose doubly-robust estimation of the
average benefit, which leads to an optimal rule even with quasi-experimental
data. \citet{mbakop_model_2021} propose a Penalized Welfare Maximization
 rule which relaxes restrictions of the policy class. 

The existing EWM literature has not addressed budget constraints with
an unknown cost. \citet{kitagawa_who_2018} consider a capacity constraint,
which they enforce using random rationing. Random rationing is not
ideal as it uses the limited resource less efficiently than accounting
for the cost of providing the welfare program to an individual. When
there is no restriction on the functional form of the eligibility
policy, \citet{bhattacharya_inferring_2012} demonstrate that given
a capacity constraint, the optimal eligibility policy is based
on a threshold on the benefit of the welfare program to an individual.
When the cost of providing the welfare program to an individual is
heterogeneous, however, budget constraints can be more complicated
than capacity constraints, and require estimation. \citet{carneiro_optimal_2020}
considers the optimal choice of covariate collection in order to maximize
the precision of the estimation for the average treatment effect.
While they also consider a constrained decision problem, the budget
constraint can be verified directly. They also allow for a more complicate
trade-off between additional covariates and additional observations.
\citet{sun_treatment_2021} propose a framework for estimating the
optimal rule under a budget constraint when there is no functional
form restriction. The main contribution of this paper is to characterize theoretical properties of statistical rules when allowing both functional form restrictions
and budget constraints with an unknown cost. 

As explained in a prior version of this paper \citep{sun2021empiricalwelfaremaximizationconstraints}, the current setting of EWM with constraint  shares the same mathematical structure as another important setting:  fairness constraints across sensitive subgroups. \cite{viviano2024fair} cast welfare of sensitive subgroups as multiple objective functions for policymakers, and solve the constrained optimal policy via the Pareto frontier. \cite{kock2024regularizing} consider a penalized objective function that penalizes violations to the constraint.  The trade-off rule proposed in this paper takes a similar form by penalizing budget overrun, but is tailored to linear objectives and constraints.
\section{\label{sec:Motivations-and-setup}Theoretical results}

\subsection{Motivation and setup}

I begin by setting up a general constrained optimization problem,
which depends on the following attributes of an individual:
\begin{equation*}
A=(\tau,C,X)\in\mathcal{A}\subseteq\mathbb{R}^{2+p}.\label{eq:population attributes}
\end{equation*}
Here $\tau$ is benefit from the treatment for the individual, $C$ is the cost
to the policymaker of providing the individual with the treatment, and $X\in\mathcal{X}\subset\mathbb{R}^{p}$
denotes their $p$-dimensional characteristics. The individual belongs
to a population that can be characterized by the joint distribution
$P$ on the attributes $A$. The unknown distribution $P\in\mathcal{P}$
is from a class of distributions $\mathcal{P}$.

A policy $g(X)\in\{0,1\}$ determines the treatment status for an
individual with observed characteristics $X$, where 1 is treatment
and 0 is no treatment. Let $\mathcal{G}$ denote the class of policies
policymakers can choose from. The  optimization problem is to find
a policy with maximal benefit while subject to a constraint on its
cost:\footnote{Following \citet{kitagawa_who_2018}, I implicitly assume the maximizer
exists in $\mathcal{G}$ with the notation in (\ref{eq:population primal}). } 
\begin{align}
\max_{g\in\mathcal{G}}\E_{P}[\tau\cdot g(X)]\text{ subject to }\E_{P}[C\cdot g(X)] & \leq k.\label{eq:population primal}
\end{align}
 If the policymaker does not have a fixed budget but still wants to
account for cost, the scalar $\tau$ can be the difference in benefit
and cost. I impose a harsh constraint at known $k$ to model a fixed budget. 

The benefit-cost attributes $(\tau,C)$ of any given individual
may be unobserved in practice. 
The focus of this paper is the setting
where policymakers can construct their estimates $(\tau_{i}^{\ast},C_{i}^{\ast})$
in a random sample of sample size $n$ along with the characteristics
$X_{i}$ from an experiment or quasi-experiment that satisfy Assumption
\ref{cond:known ps Donsker}, as discussed later.

Applying the Law of Iterated Expectation, the constrained optimization
problem (\ref{eq:population primal}) can be written as
\begin{equation*}
\max_{g\in\mathcal{G}}\E_{P}[\E_{P}[\tau\mid X]\cdot g(X)]\text{ subject to }\E_{P}[\E_{P}[C\mid X]\cdot g(X)]\leq k.
\end{equation*}
When the eligibility policy can be based on any characteristics
whatsoever, the class of available policies is unrestricted i.e. $\mathcal{G}=2^{\mathcal{X}}$.
In this unrestricted class, when the cost is non-negative, the above
expression makes clear that the optimal eligibility policy is based
on thresholding by the benefit-cost ratio $\E_{P}[\tau\mid X]/\E_{P}[C\mid X]$
where the numerator and the denominator are respectively the average
effects conditional on the observed characteristics (CATE) and the
conditional average resource required. Appendix \ref{sec:Primitive-assumptions-and}
provides a formal statement. 

Given a random sample, to approximate the optimal eligibility policy
$g_{P}^{\ast}$, one  can estimate the benefit-cost ratio based on
the estimated CATE and the estimated conditional average resource
required. The resulting statistical rule selects eligibility policies
that are thresholds based on the estimated benefit-cost ratio. The
challenge is that the selected eligibility policy can be hard to
implement when the estimated benefit-cost ratio is a complicated function
of $X$. Restrictions on the policy class $\mathcal{G}$ address
this issue. A common restriction is to consider thresholds based directly
on $X$, e.g. assigning eligibility when an individual's income is
below a certain value. 

Restrictions on the policy class $\mathcal{G}$ mean that there
might not be closed-form solutions to the population problem (\ref{eq:population primal}).
In particular, the constrained optimal eligibility policy $g_{P}^{\ast}$
might not be an explicit function of the CATE and the conditional
average resource required. Therefore it might be difficult to directly
approximate $g_{P}^{\ast}$ based on the estimated CATE and the estimated
conditional average resource required. However, this is not an obstacle
to deriving guarantees for the statistical  rules. As I demonstrate later,
the derivation does not require the knowledge of the functional form
of the constrained optimal eligibility policy $g_{P}^{\ast}$. 

I next specialize the constrained optimization problem to selecting
eligibility policy for welfare programs with the example of Medicaid
expansion. In the example of implementing welfare programs, policies
take the form of eligibility policies. I restrict attention to non-randomized
policies as in the leading example of welfare programs, deterministic
policies such as income thresholds are more relevant.  Theoretically oriented readers may
proceed directly to Section \ref{subsec:Implications-of-Data-dependent}.
\begin{example}
\textbf{Welfare program eligibility policy under budget constraint\label{exa:Welfare-program-eligibility} }
\end{example}
Suppose the government wants to implement some welfare program. The
treatment in this example is eligibility for such welfare program.
Due to a limited budget, the government cannot make eligibility universal
and can only provide eligibility to a subpopulation. To use the budget
efficiently, policymakers consider the constrained optimization problem
(\ref{eq:population primal}). In this example, the policy $g(X)$
assigns an individual to eligibility based on their observed characteristics
$X$, and is usually referred to as an eligibility policy. I denote
$\tau$ to be the benefit experienced by an individual after receiving
eligibility for the welfare program. Specifically, let $(Y_{1},Y_{0})$
denote the potential outcomes that would have been observed if an
individual were assigned with and without eligibility, respectively.
The benefit from eligibility policy is therefore defined as $\tau\coloneqq Y_{1}-Y_{0}$.
Note that maximizing benefit is equivalent to maximizing the outcomes
(welfare) under the utilitarian social welfare function: $\E_{P}[Y_{1}\cdot g(X)+Y_{0}(1-g(X))]$.
I denote $C$ to be the potential cost from providing an individual
with eligibility for the welfare program. Both $\tau$ and $C$
are unobserved at the time of assignment and will need to be estimated. 

Policymakers might be interested in multiple outcomes for an in-kind
transfer program. \citet{hendren_unified_2020} capture benefits by
the willingness to pay (WTP). Assuming eligible individuals make optimal
choices across multiple outcomes, the envelope theorem allows policymakers
to focus on benefit in terms of one particular outcome. 

\paragraph{Medicaid Expansion }

Medicaid is a government-sponsored health insurance program intended
for the low-income population in the United States. Up till 2011,
many states provided Medicaid eligibility to able-bodied adults with
income up to 100\% of the federal poverty level. The 2011 Affordable
Care Act (ACA) provided resources for states to expand Medicaid eligibility
for all adults with income up to 138\% of the federal poverty level
starting in 2014. 

Suppose policymakers want to maximize the health benefit of Medicaid
by adopting a more flexible expansion policy. Specifically, they relax the uniform income threshold of
138\% and allow the income thresholds to vary with the number of
children in the household. Once the income thresholds are set, they must be codified in state legislation and publicly announced.\footnote{The legislated eligibility policy is publicly available on federal websites such as \href{https://www.macpac.gov/medicaid-101/eligibility/}{MACPAC}.} Therefore, in this example, it is reasonable to assume that the eligibility policy must be determined ex ante. 

The policy class in this example includes income thresholds that
can vary with the number of children in the household:
\begin{equation}\label{eq:thresholds}
\mathcal{G}=\left\{ g(x)=\begin{cases}
\mathbf{1}\{\text{income}\leq\beta_{0}\} & ,\ \text{numchild}=0\\
\vdots\\
\mathbf{1}\{\text{income}\leq\beta_{j}\} & ,\ \text{numchild}=j
\end{cases}\right\} 
\end{equation}
for characteristics $x=(\text{income},\ \text{numchild})$ and $\ \beta_{j}\geq0$.

Convincing policymakers to adopt a more flexible expansion policy as in~\eqref{eq:thresholds} may still require specifying a clear budget target to ensure that the new policy does not exceed the expenditure level of the current expansion with the uniform income threshold of 138\%. Given that Medicaid is a countercyclical program and has some flexibility to accommodate budget overruns, later in Section~\ref{sec:Trade-off}, I develop a new rule to allow explicit trade-off between welfare gains from additional
spending and a penalty for budget overruns.

Correspondingly, the constrained optimization problem (\ref{eq:population primal})
sets $\tau$ to be the health benefit from Medicaid, $C$ to be
the cost to Medicaid, and the appropriate threshold $k$ to be the average cost to Medicaid under the current expansion policy with the uniform income threshold of 138\%. The
characteristics $X$ include both income and number of children in
the household.

\subsection{\label{subsec:Implications-of-Data-dependent}Desirable properties
for statistical rules}

To simplify the notation, I define the \emph{welfare} function and
the \emph{budget} function:
\begin{equation*}
W(g;P)=\E_{P}[\tau\cdot g(X)],\ B(g;P)=\E_{P}[C\cdot g(X)].\label{eq:population fcn}
\end{equation*}
and the constrained optimal policy $g_{P}^{\ast}$ is therefore the solution to \begin{align*}
\max_{g\in\mathcal{G}}W(g;P)\text{ subject to }B(g;P) & \leq k. 
\end{align*}
As explained in Example \ref{exa:Welfare-program-eligibility} from
Section \ref{sec:Motivations-and-setup}, under a utilitarian social
welfare function, maximizing the benefit with respect to eligibility
 policy is equivalent to maximizing the welfare, which is why I
refer to $W(g;P)$ as the welfare function. The welfare function and
the budget function are both deterministic functions from $\mathcal{G}\rightarrow\mathbb{R}$.
The index by the distribution $P$ highlights that welfare and budget
of policy $g$ vary with $P$, and in particular, whether a policy
$g$ satisfies the budget constraint depends on which distribution
$P$ is of interest.

When the benefit-cost attributes $(\tau,C)$ are unobserved and
the distribution $P$ is unknown, both the welfare function and the
budget function are unknown functions. Denote by $\widehat{g}$ a
statistical rule that selects an eligibility policy after observing
some experimental data of sample size $n$ distributed according to
$P^{n}$. This section provides formal definitions for two desirable
properties of $\widehat{g}$.
\begin{defn}
A statistical rule $\widehat{g}$ is \emph{pointwise asymptotically
welfare-efficient} under the data distribution $P$ if for any $\epsilon>0$
\begin{equation*}
\lim_{n\rightarrow\infty}\Pr_{P^{n}}\left\{ W(\widehat{g};P)-W(g_{P}^{\ast};P)<-\epsilon\right\} =0,\label{eq:uniform conv}
\end{equation*}
and \emph{uniformly asymptotically welfare-efficient} over the class
of distributions $\mathcal{P}$ if for any $\epsilon>0$
\[
\lim_{n\rightarrow\infty}\sup_{P\in\mathcal{P}}\Pr_{P^{n}}\left\{ W_{}(\widehat{g};P)-W(g_{P}^{\ast};P)<-\epsilon\right\} =0.
\]
 A statistical rule is\emph{ pointwise asymptotically feasible }under
the data distribution $P$ if 
\[
\lim_{n\rightarrow\infty}\Pr_{P^{n}}\{B(\widehat{g};P)>k\}=0,
\]
and \emph{uniformly asymptotically} \emph{feasible} over the class
of distributions $\mathcal{P}$ if
\begin{equation*}
\lim_{n\rightarrow\infty}\sup_{P\in\mathcal{P}}\Pr_{P^{n}}\{B(\widehat{g};P)>k\}=0.\label{eq: uniform mistake}
\end{equation*}
$\ensuremath{\blacksquare}$
\end{defn}
The above two properties build on the existing EWM literature. For
the first property, the current EWM literature evaluates statistical
rules by whether they attain at least $W(g_{P}^{\ast};P)$ in expectation over repeated sample draws as $n\rightarrow\infty$.
Instead, I focus on convergence in probability,
where the probability of $\widehat{g}$ selecting eligibility policies
that achieve strictly lower welfare than $g_{P}^{\ast}$ approaches
zero as $n\rightarrow\infty$. %
In the setting of the existing EWM literature, the constrained optimal
 policy $g_{P}^{\ast}$ is also the unconstrained optimal policy
$\mathcal{G}$. Therefore it is impossible for any statistical rule
$\widehat{g}$ to select a policy that achieves higher value than
$g_{P}^{\ast}$. In my setting, however, the constrained optimal policy
$g_{P}^{\ast}$ is not necessarily the unconstrained optimal policy.
Thus, I allow the statistical rule $\widehat{g}$ to select a policy
that achieves higher welfare than $g_{P}^{\ast}$ for all data distributions,
albeit at the cost of violating the budget constraint. 

The second property is new to the EWM literature. It imposes that
given a large enough sample size, the statistical rule $\widehat{g}$
is unlikely to select infeasible eligibility policies that violate
the budget constraint, so that it is ``asymptotically feasible''. %
Asymptotic feasibility of statistical rules is specific to the current
setting where the budget constraint involves unknown cost. Exactly satisfying a fixed budget constraint without the smallest violation is the most conservative way to articulate policy makers' preferences.

While both are desirable properties, the next section shows a negative
result that it is impossible for a statistical rule to satisfy both
properties when the data distribution is unknown and belongs to a
sufficiently rich class of distributions $\mathcal{P}$.

\subsection{\label{subsec:Mistakes-from-any}An impossibility result}

Uniformity is a desirable property for a statistical rule, ensuring that its performance guarantees hold uniformly over a class of DGPs, thereby providing robustness to uncertainty about the true DGP. However, in this section, I prove  an  impossibility result that no statistical rule can be both uniformly asymptotically welfare-efficient and uniformly asymptotically feasible in a sufficiently rich class of distributions $\mathcal{P}$ described in the following Assumptions \ref{cond:conv constraint}-\ref{assu:gain when exactly binding}. The intuition is that there exists a sequence $P_{h_n}$ that converges to $P_0\in \mathcal{P}$, but along the sequence, their corresponding constrained optimum $g^{\ast}_{h_n}$ does not converge to the constrained optimum $g^{\ast}_{P_0}$ under $P_0$.  Assumptions \ref{assu:exactly binding optimal} and \ref{assu:gain when exactly binding} characterize such point of discontinuity $P_0$.

\begin{assumption}
\label{cond:conv constraint}\emph{Contiguity. }There exists a distribution
$P_{0}\in\mathcal{P}$ under which the set $\mathcal{G}_{0}=\left\{ g:B(g;P_{0})=k\right\} $ of eligibility
policies satisfying the constraint exactly is non-empty.
Furthermore, the class of distributions $\mathcal{P}$ includes a
sequence of data distributions $\{P_{h_{n}}\}$ contiguous to $P_{0}$,
under which for all $g\in\mathcal{G}_{0}$, there exists some $c>0$
such that
\[
\sqrt{n}\cdot\left(B(g;P_{h_{n}})-k\right)>c.
\]
\end{assumption}
Assumption \ref{cond:conv constraint}
assumes there exists a sequence of data distributions $\{P_{h_{n}}\}$ that differ only marginally relative to $P_{0}$. Moreover, the budget function $B(g;P_{h_{n}})$, evaluated at polices that meet the budget constraint exactly under $P_{0}$, approaches $B(g;P_{0})$ from above. In Appendix \ref{subsec:Primitive-for-contiguity}, I give
more primitive assumptions under which Assumption \ref{cond:conv constraint}
is guaranteed to hold, requiring the sequence to be differentiable in quadratic mean at $P_{0}$, and along the sequence, the budget function $B(g;P_{h_n})$ is twice continuously differentiable at $P_{0}$ with positive derivatives. 
These primitive assumptions are relatively weak, and have also been considered in the literature to construct the local parametrization around $P_{0}$, e.g. \cite{hirano_asymptotics_2009}. 
\begin{assumption}
\label{assu:exactly binding optimal}\emph{Binding constraint.} Under
the data distribution $P_{0}$, the constraint is satisfied exactly
at the constrained optimum i.e. $B(g_{P_{0}}^{\ast};P_{0})=k$. 
\end{assumption}
\begin{assumption}
\label{assu:gain when exactly binding}\emph{Separation.} Under the
data distribution $P_{0}$, $\exists\epsilon>0$ such that for any
feasible policy $g$, whenever 
\[
\left|B(g;P_{0})-B(g_{P_{0}}^{\ast};P_{0})\right|>0,
\]
we have 
\[
W(g_{P_{0}}^{\ast};P_{0})-W(g;P_{0})>\epsilon.
\]
Equivalently, \emph{$W(g_{P_{0}}^{\ast};P_{0})$ is separated from
that of other feasible policies with different $B(g;P_{0})$.} 
\end{assumption}

Consider a one-dimensional policy class $\mathcal{G}$, e.g. income
thresholds. Figure \ref{fig:Discontinuous-limit-distribution} illustrates
a distribution $P_{0}$ that satisfies both Assumptions \ref{assu:exactly binding optimal}
and \ref{assu:gain when exactly binding}, while both the welfare
function $W_{}(g;P_{0})$ and the budget function $B(g;P_{0})$ are
still continuous in $g$, satisfying Assumption~\ref{cond:conv constraint}.  Importantly,  the constrained optimal
 policy $g_{P_{0}}^{\ast}$ satisfies the budget constraint exactly,
i.e. $B(g_{P_{0}}^{\ast};P_{0})=k$, but is separated from the rest
of feasible eligibility policies such that there exists a neighborhood
around $g_{P_{0}}^{\ast}$ where feasible policies can achieve welfare
gain without any effect on the budget. 

\begin{thm}
\emph{\label{thm: no uniform estimator}} Suppose Assumption \ref{cond:conv constraint}
holds for the class of data distributions $\mathcal{P}$. For $P_{0}\in\mathcal{P}$
considered in Assumption \ref{cond:conv constraint}, suppose it also
satisfies Assumptions \ref{assu:exactly binding optimal} and \ref{assu:gain when exactly binding}.
Then no statistical rule can be both uniformly asymptotically welfare-efficient
and uniformly asymptotically feasible. In particular, if a statistical
rule $\widehat{g}$ is pointwise asymptotically welfare-efficient
and pointwise asymptotically feasible under $P_{0}$, then it is not
uniformly asymptotically feasible.
\end{thm}
Figure \ref{fig:Discontinuous-limit-distribution} provides some intuition
for Theorem \ref{thm: no uniform estimator}.  Note that if a statistical rule
$\widehat{g}$ is pointwise asymptotically welfare-efficient and pointwise
asymptotically feasible under $P_{0}$, then it has to select eligibility
policies close to $g_{P_{0}}^{\ast}$ with high probability over repeated
sample draws distributed according to $P_{0}^{n}$ as $n\rightarrow\infty$. 

Under Assumption \ref{cond:conv constraint}, the class of distributions
$\mathcal{P}$ is sufficiently rich so that along a sequence of data
distributions $\{P_{h_{n}}\}$ that is contiguous to $P_{0}$ as $n\rightarrow\infty$,
the budget functions $B(g;P_{h_{n}})$ converge to $B(g;P_{0})$ while
$B(g_{P_{0}}^{\ast};P_{h_{n}})>k$, i.e. $g_{P_{0}}^{\ast}$ is not
feasible under $P_{h_{n}}$. Figure \ref{fig:Discontinuous-limit-distribution}
showcases $P_{1}$ as one distribution from this sequence. The contiguity
between $\{P_{h_{n}}\}$ and $P_{0}$ implies that the statistical
rule $\widehat{g}$ must select policies close to $g_{P_{0}}^{\ast}$
with high probability under $P_{h_{n}}^{n}$ as well. However, the
 policy $g_{P_{0}}^{\ast}$ is infeasible under $P_{h_{n}}$, and
therefore the statistical rule $\widehat{g}$ cannot be asymptotically
feasible under $P_{h_{n}}$. 

\begin{figure}[H]
\begin{centering}
\caption{\label{fig:Discontinuous-limit-distribution}Illustration for Assumptions
\ref{cond:conv constraint}-\ref{assu:gain when exactly binding}
underlying Theorem \ref{thm: no uniform estimator}}
 \includegraphics[scale=0.3]{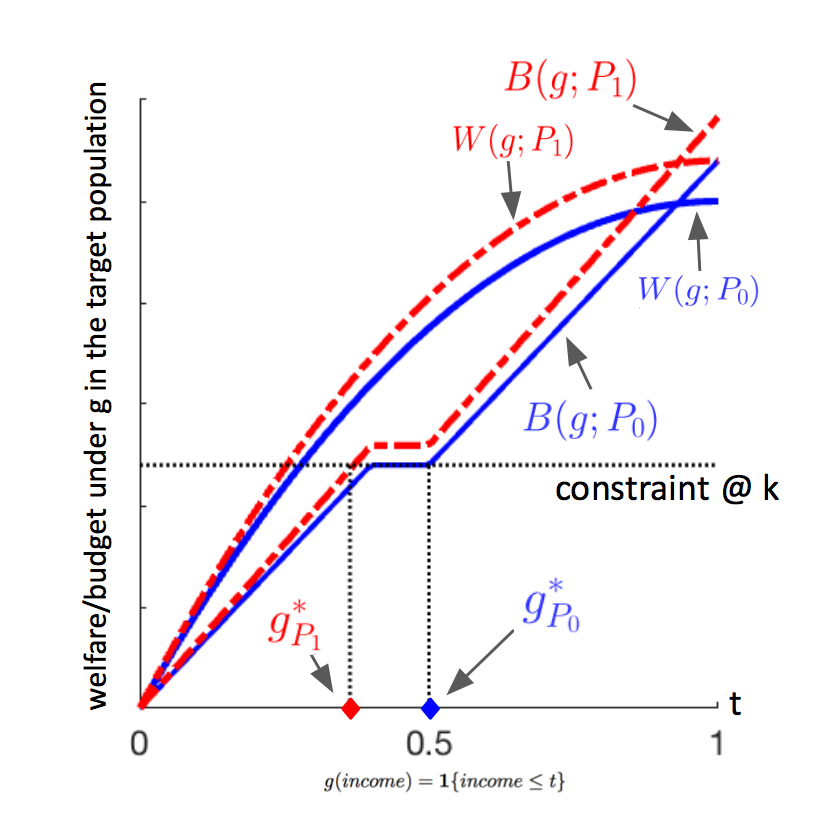}
\par\end{centering}
\emph{Notes}: This figure plots welfare functions $W_{}(g;P)$ and
budget functions $B(g;P)$ for populations distributed according to
$P_{0}$ (blue solid lines) or $P_{1}$ (red dashed lines), where
$P_{1}$ is a distribution from the sequence of distributions $\{P_{h_{n}}\}$
that is contiguous to $P_{0}$ under Assumption \ref{cond:conv constraint}.
The distribution $P_{0}$ satisfies Assumptions \ref{assu:exactly binding optimal}
and \ref{assu:gain when exactly binding}. The $x$-axis indexes a
one-dimensional policy class $\mathcal{G}=\{g:g(x)=\mathbf{1}\{x\leq t\}\}$
for a one-dimensional continuous characteristic $X_{i}$ with support
on $[0,1]$, e.g. eligibility policies based on income thresholds.
The black dotted line marks the budget threshold $k$. The bold blue
dot marks $g_{P_{0}}^{\ast}$, the constrained optimal eligibility
 policy under $P_{0}$. The bold red dot marks $g_{P_{1}}^{\ast}$,
the constrained optimal eligibility policy under $P_{1}$. 
\end{figure}

Figure~\ref{fig:Discontinuous-limit-distribution} also highlights the importance of Assumptions~\ref{assu:exactly binding optimal} and~\ref{assu:gain when exactly binding} in driving the impossibility result. Importantly, under $P_0$, the budget constraint is binding at the constrained optimal threshold (Assumption~\ref{assu:exactly binding optimal}), and increasing the eligibility threshold here  has a strictly positive impact on welfare but zero
impact on budget (Assumption~\ref{assu:gain when exactly binding}). If, instead, increasing the threshold strictly raises both welfare and budget, then Assumption~\ref{assu:gain when exactly binding} is violated. In Section~\ref{subsec:mistake-uniformity}, I demonstrate that relaxing either Assumption~\ref{assu:exactly binding optimal} or~\ref{assu:gain when exactly binding} gives rise to statistical rules that achieve uniformity within certain subclasses of DGPs.
Nonetheless, the impossibility result remains relevant for some real-world policy settings where such assumptions may approximately hold. While the nominal cost of welfare program eligibility may be positive, the Marginal Value of Public Funds (MVPF) framework emphasizes that the relevant cost is the net cost, which incorporates fiscal externalities, such as increased tax revenue if the program increases individuals' incomes. Specifically,   \citet{hendren_unified_2020} estimate fourteen welfare programs (out of 133) to have negative or zero net cost to the government, which implies these programs ``pay for themselves'', aligning with Assumptions \ref{assu:exactly binding optimal} and \ref{assu:gain when exactly binding}.  As a result, the impossibility result suggests that statistical rules designed for such cases may be highly sensitive to small changes in the underlying distributions even in large samples.\footnote{As estimated by \citet{hendren_unified_2020}, expanding Medicaid, as in  Example \ref{exa:Welfare-program-eligibility}, entails a strictly positive net cost.   The impossibility result discussed in this section therefore does not apply to this example. Instead, I use this example to illustrate a new trade-off rule proposed later in Section~\ref{sec:Trade-off}.}

\section{\label{subsec:Mistakes-from-harsh}Sample-analog
rule }

The previous negative result implies that no statistical rule can
be both uniformly asymptotically welfare-efficient and uniformly asymptotically
feasible. Thus, policymakers might want to consider statistical rules
that satisfy one of these two properties. A direct extension to the existing approach in the EWM literature,
the sample-analog rule, might be a natural candidate. In this section, I show the direct extension is neither uniformly asymptotically welfare-efficient nor uniformly asymptotically
feasible.%

I first describe the EWM approach \citep{kitagawa_who_2018} and its
direct extension. 
Since $(\tau,C)$ involves potential outcomes, they are often unobserved
and require estimation based on RCT that introduces estimation errors
in addition to sampling errors. Section \ref{subsec:Estimates-for-benefit} describes how to construct individuals' benefit and cost estimates $(\tau_{i}^{\ast},C_{i}^{\ast})$. To highlight the drawback of the direct
extension to EWM,  I first consider settings where we observe
an experimental data of sample size $n$ where $(\tau,C)$ is directly
observable, i.e. $(\tau_{i}^{\ast},C_{i}^{\ast})=(\tau_{i},C_{i})$.
The goal of the simplification is to highlight that the non-uniformity
I show below can arise from sampling errors alone. 

One can estimate the welfare function and the budget function using
their sample-analog versions:
\begin{equation}
\widehat{W}_{n}(g)\coloneqq\frac{1}{n}\sum_{i}\tau_{i}^{\ast}\cdot g(X_{i}),\ \widehat{B}_{n}(g)\coloneqq\frac{1}{n}\sum_{i}C_{i}^{\ast}\cdot g(X_{i}).\label{eq:sample analogs}
\end{equation}

A direct extension to the existing approach in the EWM literature
is a statistical rule that solves the sample version of the population
constrained optimization problem (\ref{eq:population primal}):
\begin{equation}
\widehat{g}_{\text{sample}}\in\arg\max_{\widehat{B}_{n}(g)\leq k}\ \widehat{W}_{n}(g).\label{eq:EWM}
\end{equation}
The subscript \textquotedblleft sample'' emphasizes how this approach
verifies whether a policy satisfies the constraint by comparing
the \emph{sample analog} $\widehat{B}_{n}(g)$ with $k$ directly,
i.e. imposes a sample-analog constraint. If no policy satisfies
the constraint, then I set $\widehat{g}_{\text{sample}}$ to not assign
any eligibility, i.e. $\widehat{g}_{\text{sample}}(x)=0$ for all
$x\in\mathcal{X}$. 

\subsection{Welfare inefficiency of the sample-analog rule}
A key insight from \citet{kitagawa_who_2018} is that without a constraint, the sample-analog rule is uniformly asymptotically welfare-efficient.
Unfortunately this intuition does not extend to the current setting where the constraint involves an unknown cost. 
There are common data distributions under which the amount of welfare loss does not
vanish even as the sample size gets larger.
Consider a one-dimensional
 policy class $\mathcal{G}=\{g:g(x)=\mathbf{1}\{x\leq t\}\}$, which
is based on thresholds of a one-dimensional continuous characteristic
$X$. Suppose the policymakers know benefit is positive for everyone
so that welfare function is strictly increasing, and only need to
estimate whether a given threshold satisfies a capacity constraint
due to imperfect take-up. Furthermore, suppose the experiment sample
observes take-up $C_{i}$ so that the only uncertainty arises from
sampling errors. Proposition \ref{prop: EWM no ptwise mistake} shows settings 
 with some zero take-ups satisfy Assumptions \ref{assu:exactly binding optimal} and \ref{assu:gain when exactly binding}, and the sample-analog rule is neither pointwise asymptotically welfare efficient nor pointwise asymptotically feasible.%
\begin{prop}
\label{prop: EWM no ptwise mistake} \emph{One-dimensional
threshold and imperfect take-up.} Consider a one-dimensional policy class $\mathcal{G}=\{g:g(x)=\mathbf{1}\{x\leq t\}\}$,
which includes thresholds for a one-dimensional continuous characteristic
$X$. Consider the special case where
under distribution $P$, the benefit $\tau>0$ almost surely and
the cost $C\in\{0,1\}$ is binary. Suppose $C=0$ for $X\in[\underline{t},\overline{t}]$, but $\Pr_{P}\{C=1\mid X\}\in(0,1)$ otherwise.  Suppose further the budget constraint is at $k = \E_{P} [C \cdot \mathbf{1}\{X \leq \underline{t}\}].$ Then $P$ satisfies Assumptions \ref{assu:exactly binding optimal} and \ref{assu:gain when exactly binding}. %
As $n\rightarrow\infty$, the sample-analog
rule $\widehat{g}_{\text{sample}}$ violates the budget constraint with probability approaching 50\%, and there exists an  $\epsilon>0$ such that it  incurs at least an  $\epsilon$ welfare loss with probability approaching 50\%:
\[
\Pr_{P^{n}}\left\{ W(g_{P}^{*};P)-W(\widehat{g}_{\text{sample}};P)\geq\epsilon\right\} \rightarrow50\%.
\]
\end{prop}

\begin{figure}[H]
\begin{centering}
\caption{\label{fig:sample-analog-illustration}Illustration for Proposition
\ref{prop: EWM no ptwise mistake}}
\subfloat[infeasible]{\begin{centering}
\includegraphics[scale=0.25]{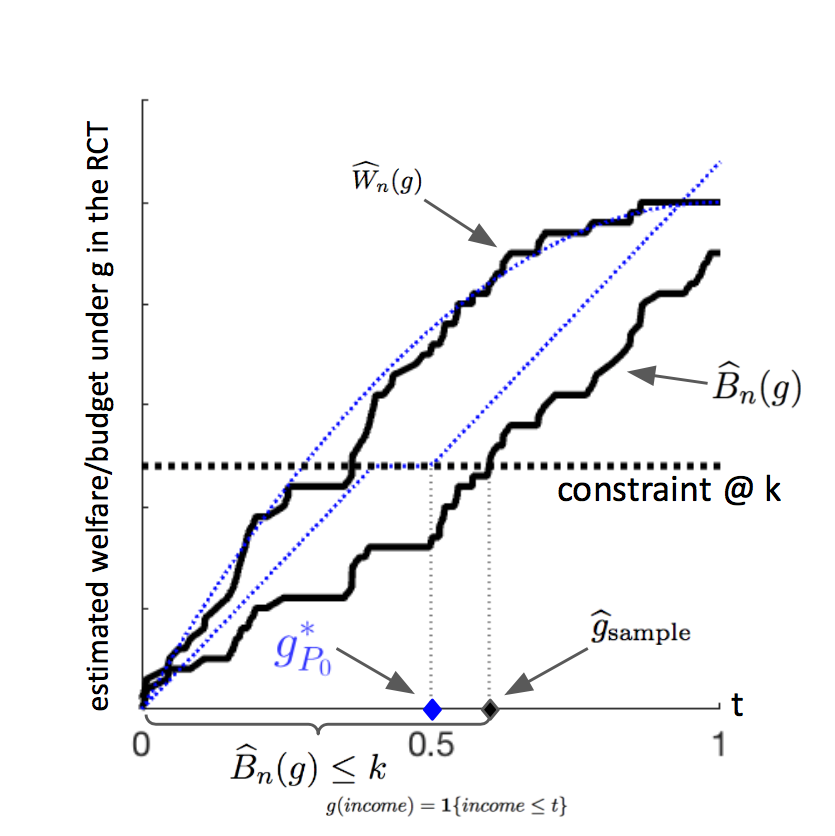}
\par\end{centering}
}\subfloat[suboptimal]{\begin{centering}
\includegraphics[scale=0.25]{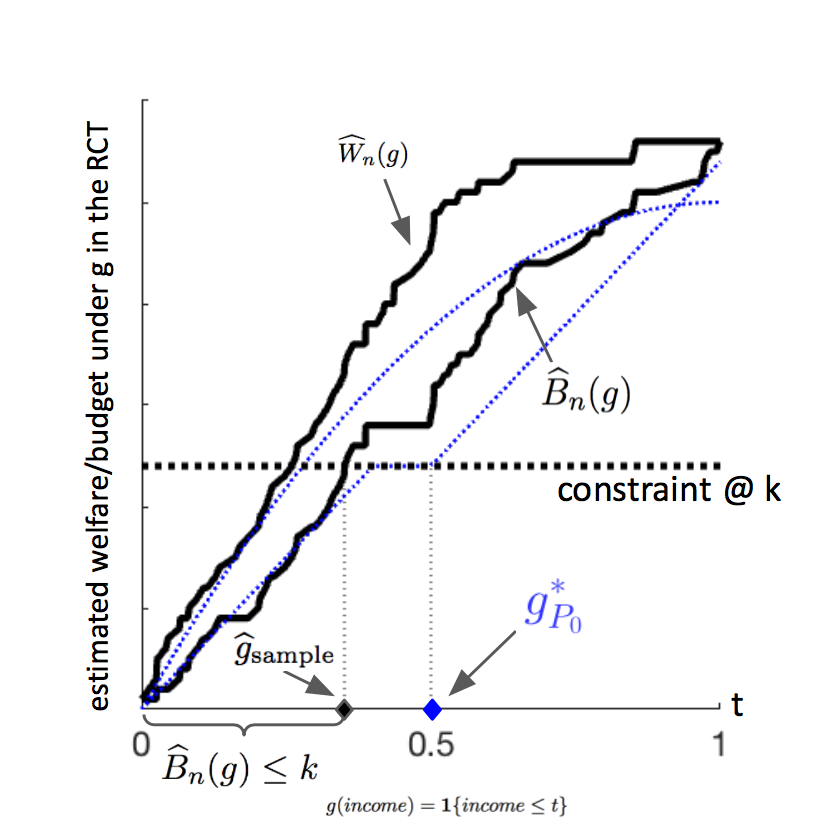}
\par\end{centering}
}
\par\end{centering}
\emph{Notes}: This figure plots the budget function $B(g;P)$ (blue
dotted line) and its sample-analogs $\widehat{B}_{n}(g)$ (black wriggly
line) based on two different observed samples in panel (a) and (b)
respectively. The $x$-axis indexes a one-dimensional policy class
$\mathcal{G}=\{g:g(x)=\mathbf{1}\{x\leq t\}\}$ for a one-dimensional
continuous characteristic $X_{i}$ with support on $[0,1]$. The black
dotted line marks the budget threshold $k$. The constrained optimal
threshold $g_{P}^{\ast}$ is $t=0.5$. The sample-analog rule $\widehat{g}_{\text{sample}}$
selects an infeasible threshold in panel (a) and selects a suboptimal
threshold in panel (b).
\end{figure}

Figure \ref{fig:sample-analog-illustration} illustrates the setup
of Proposition \ref{prop: EWM no ptwise mistake}, where the sampling
uncertainty can be particularly problematic for the sample-analog
rule. Since policymakers know the benefit is positive for everyone,
$\widehat{g}_{\text{sample}}$ takes a simple form of the highest
threshold where the sample-analog constraint is satisfied exactly.
The driving force behind the failure of $\widehat{g}_{\text{sample}}$
as described in Proposition \ref{prop: EWM no ptwise mistake} is
that due to sampling uncertainty, whether a policy satisfies the
sample-analog constraint is an imperfect measure of whether it satisfies
the constraint in the population. 

Since the sample-analog rule $\widehat{g}_{\text{sample}}$ restricts
attention to policies that satisfy the sample-analog constraint, there
is no guarantee the selected policy is actually feasible. This
is very likely to happen when there is welfare gain in exceeding the
budget constraint as in the setup of Proposition \ref{prop: EWM no ptwise mistake}
where $W(g;P)$ is strictly increasing in $g$. Therefore, the sample-analog
rule $\widehat{g}_{\text{sample}}$ is not asymptotically feasible
under $P$. As illustrated in Figure \ref{fig:sample-analog-illustration},
after observing a sample depicted in panel (a), the sample-analog
rule picks an infeasible threshold because the sample-analog constraint
is still satisfied there. However, as argued later in Corollary~\ref{prop:analog violations}, the probability of large budget violations vanishes as the sample size increases.

The more problematic case is illustrated in Figure \ref{fig:sample-analog-illustration} panel (b), where the sample-analog rule picks a suboptimal threshold because the sample-analog constraint is violated at the constrained optimum $g_{P}^{\ast}$. 
In the setup of Proposition \ref{prop: EWM no ptwise mistake}, when the sample-analog rule $\widehat{g}_{\text{sample}}$ misses $g_{P}^{\ast}$, it is guaranteed to select a suboptimal policy and therefore it is  welfare-inefficient under $P$ even asymptotically. This result relies on the assumption that the welfare function strictly increases when the budget function remains constant in the neighborhood of the budget constraint. This aligns with Assumption  \ref{assu:gain when exactly binding}, which, as discussed in Section~\ref{subsec:Mistakes-from-any}, reflect real-world scenarios where some welfare programs have zero or negative net cost to the government. I demonstrate that the same issue extends to broader contexts by employing a simulation calibrated to a real-world DGP from the OHIE in Section~\ref{sec:Simulation-study}. If instead  the budget function is strictly increasing and violates Assumption~\ref{assu:gain when exactly binding},  I show the sample-analog rule is asymptotically welfare-efficient in Proposition~\ref{prop: EWM ptwise correct} of Appendix~\ref{appx:addl theory}.

\subsection{Estimates for benefit and cost\label{subsec:Estimates-for-benefit}}

In the remainder of the paper, I present constructive results that modify the sample-analog rule and propose a new rule. To lay the groundwork, this section outlines the benefit and cost estimates.  The appropriate expressions for these estimates depend on the type
of observed data. Below I state the estimates formed based an RCT
that randomly assigns the eligibility, which is the leading case of
\citet{kitagawa_who_2018}. The observed data $\{A_{i}^{\ast}\}_{i=1}^{n}$
consists of i.i.d. observations $A_{i}^{\ast}=(Y_{i},Z_{i},D_{i},X_{i})\in\mathcal{A}^{\ast}$.
The distribution of $A_{i}^{\ast}$ is induced by the distribution
of $(Y_1,Y_0,C,X)$ as in the population, as well as the sampling design
of the RCT. Here $D_{i}$ is an indicator for being in the eligibility
arm of the RCT, $Y_{i}$ is the observed outcome and $Z_{i}$ is the
observed cost of providing eligibility to an individual participating
in the RCT. The observed cost is mechanically zero if an individual
is not randomized into the eligibility arm. The estimates for $(\tau,C)$
are 
\begin{equation}
\tau_{i}^{\ast}=\alpha(X_{i},D_{i})\cdot Y_{i},\ C_{i}^{\ast}=\frac{D_{i}}{p(X_{i})}\cdot Z_{i},\label{eq:IPW}
\end{equation}
where $\alpha(X_{i},D_{i})=\frac{D_{i}}{p(X_{i})}-\frac{1-D_{i}}{1-p(X_{i})}$
and $p(X_{i})$ is the propensity score, the probability of receiving
eligibility conditional on the observed characteristics. Since the
sampling design of an RCT is known, the propensity score is a known
function of the observed characteristics.

\begin{assumption}
\label{cond:known ps Donsker} \emph{Estimation quality. }The recentered
empirical processes $\widehat{W}_{n}(\cdot)$ and $\widehat{B}_{n}(\cdot)$
defined in (\ref{eq:sample analogs}) converge to mean-zero Gaussian
processes $G_{P}^{W}$ and $G_{P}^{B}$ uniformly over $g\in\mathcal{G}$,
with covariance functions $\Sigma_{P}^{W}(\cdot,\cdot)$ and $\Sigma_{P}^{B}(\cdot,\cdot)$
respectively:
\begin{align*}
\left\{ \sqrt{n}\cdot\left(\frac{1}{n}\sum_{i}\tau_{i}^{\ast}\cdot g(X_{i})-W(g;P)\right)\right\} _{g\in\mathcal{G}} & \rightarrow_{d}G_{P}^{W}\\
\left\{ \sqrt{n}\cdot\left(\frac{1}{n}\sum_{i}C_{i}^{\ast}\cdot g(X_{i})-B(g;P)\right)\right\} _{g\in\mathcal{G}} & \rightarrow_{d}G_{P}^{B}
\end{align*}
Moreover, the convergence holds uniformly over $P\in\mathcal{P}$.
The covariance functions are uniformly bounded, with diagonal entries
bounded away from zero uniformly over $g\in\mathcal{G}$. There is
a uniformly consistent estimator $\widehat{\Sigma}^{B}(\cdot,\cdot)$
of the covariance function $\Sigma_{P}^{B}(\cdot,\cdot)$.
\end{assumption}
Appendix \ref{subsec:Results-with-estimated} gives primitive
assumptions under which Assumption \ref{cond:known ps Donsker} is
guaranteed to hold for $\widehat{W}_{n}(\cdot)$ and $\widehat{B}_{n}(\cdot)$
constructed using an RCT such as in (\ref{eq:IPW}) or an observational
study, assuming unconfoundedness
and strong overlap. As standard in the literature, I need to restrict the complexity
of the policy class $\mathcal{G}$. The policy class of income thresholds considered in this paper is in fact a rather simple class with VC-dimension $d+1$ where $d$ is the number of different thresholds as in (\ref{eq:thresholds}). 

\subsection{\label{sec:Controlling-the-Mistake}Modification that ensures feasibility}
Let $\hat{\mathcal{G}}=\{g\in\mathcal{G}:\widehat{B}_{n}(g)\leq k\}$
denote the set of policies that $\widehat{g}_{\text{sample}}$ can
choose from, which contains policies that do not violate the sample-analog
of the budget constraint. Unsurprisingly, for any finite sample, $\hat{\mathcal{G}}$ can always contain infeasible policies, sometimes of sizable budget violations. The probability
that $\hat{\mathcal{G}}$ contains a policy that violates
the population budget constraint by a fixed amount $c$ is as follows:
\begin{equation}
\Pr_{P^{n}}\left\{ \exists g:g\in\hat{\mathcal{G}}\text{ and }B(g;P)>k+c\right\} \label{eq:sample-analog-violation}
\end{equation}

The corollary stated below shows the chance that a large amount of
budget violation of $c>0$ occurs is smaller when there is less variability
in $\widehat{B}_{n}(g)$. The chance also vanishes to zero as the
sample size gets larger. %
\begin{cor}
\label{prop:analog violations}Under Assumption \ref{cond:known ps Donsker},
an upper bound of the probability (\ref{eq:sample-analog-violation}) can be approximated by the CDF of $\underset{g\in\mathcal{G}}{\inf}\frac{G_{P}^{B}(g)}{\Sigma_{P}^{B}(g,g)^{1/2}}$
evaluated at $\frac{-\sqrt{n}\cdot c}{\max_{g\in\mathcal{G}}\Sigma^{B}(g,g)^{1/2}}$. %
\end{cor}

A simple modification to the sample-analog rule can reduce the probability of selecting infeasible policies. Instead of $\hat{\mathcal{G}}$, let $\gmistake$ choose policies from a subset $\hat{\mathcal{G}}_{\alpha}$ defined in Theorem~\ref{thm: uniform size} and maximize $\widehat{W}_n(g)$ as before. Then as a direct consequence of Theorem \ref{thm: uniform size}, with
probability at least $1-\alpha$ this modified rule is guaranteed
to not mistakenly choose infeasible eligibility policies. In practice, $\alpha$ may be set at the conventional level, e.g. 5\%. Then effectively this modification first forms a uniform upper confidence band for the costs of all the policies with asymptotic coverage at least 95\% and $\hat{\mathcal G}_{\alpha}$ collects only the policies for which the upper bound on cost is below the threshold.
\begin{thm}
\emph{\label{thm: uniform size}}Suppose Assumption \ref{cond:known ps Donsker}
holds for the class of data distributions $\mathcal{P}$. Collect
eligibility policies in 
\begin{equation}
\hat{\mathcal{G}}_{\alpha}=\left\{ g:g\in\mathcal{G}\text{ and }\frac{\sqrt{n}\left(\widehat{B}_{n}(g)-k\right)}{\widehat{\Sigma}^{B}(g,g)^{1/2}}\leq c_{\alpha}\right\} ,\label{eq:test inversion}
\end{equation}
where $c_{\alpha}$ is the $\alpha$-quantile from $\underset{g\in\mathcal{G}}{\inf}\frac{G_{P}^{B}(g)}{\Sigma_{P}^{B}(g,g)^{1/2}}$
for $G_{P}^{B}$ the Gaussian process defined in Assumption \ref{cond:known ps Donsker},
and $\widehat{\Sigma}^{B}(\cdot,\cdot)$ the consistent estimator
for its covariance function. Then 
\begin{equation*}
\limsup_{n\rightarrow\infty}\sup_{P\in\mathcal{P}}\Pr_{P^{n}}\{\hat{\mathcal{G}}_{\alpha}\cap\mathcal{G}^{+}\neq\emptyset\}<\alpha,\label{eq:asym size}
\end{equation*}
where $\mathcal{G}^{+}=\{g:B(g;P)>k\}$ is the set of infeasible eligibility
policies.
\end{thm}
Note that in \eqref{eq:test inversion} the sample-analog constraint is tightened by $c_{\alpha}\cdot\frac{\widehat{\Sigma}^{B}(g,g)^{1/2}}{\sqrt{n}}$
where $c_{\alpha}$ is negative, which means the class $\hat{\mathcal{G}}_{\alpha}$
only includes eligibility policies where the constraint is slack in
the sample. The sample-analog constraint is tightened proportionally
to the standard error to reflect that $\widehat{B}_{n}(g)$ might
be particularly noisy for some $g$. The tightening therefore shrinks inversely proportional to
the (square root of) sample size because intuitively larger sample
size reduces the sampling uncertainty.

\subsection{Modification that ensures uniformity in subclasses of DGPs}\label{subsec:mistake-uniformity}
Although $\gmistake$ achieves  asymptotic feasibility uniformly over a wide class of DGPs, it may lead to welfare inefficiencies under certain DGPs. Below I show that if   one lets $\alpha_n\rightarrow0$ as sample size increases, at a rate such that $c_{\alpha_{n}}=o(n^{1/2})$, then the modified sample-analog rule $\widehat{g}_{\alpha_{n}}$ is both asymptotically welfare-efficient
and  asymptotically feasible uniformly over two subclasses of distributions. 
\begin{cor}
\label{cor:optimal mistake control}Consider a family of distributions $\mathcal{P}$ such that for all $P \in \mathcal{P}$, 
Assumption~\ref{cond:known ps Donsker} holds and the budget constraint is uniformly slack under the constrained optimal policy. 
That is, there exists $\delta > 0$ such that 
\[
B(g_P^{\ast}; P) \leq k - \delta \quad \text{for all } P \in \mathcal{P},
\]
so that Assumption~\ref{assu:exactly binding optimal} is violated for every $P \in \mathcal{P}$. Then the modified sample-analog rule\begin{equation}
\widehat{g}_{\alpha_{n}}\in\arg\max_{g\in\hat{\mathcal{G}}_{\alpha_{n}}}\widehat{W}_{n}(g)\label{eq:EWM-size}
\end{equation}
is  asymptotically welfare-efficient and asymptotically feasible uniformly under $\mathcal{P}$ for $\alpha_{n}\rightarrow0$ at a rate such that $c_{\alpha_{n}}=o(n^{1/2})$.
\end{cor}

\begin{cor}
\label{cor:positive cost}Let the policy class $\mathcal{G}$ be based on continuously distributed
covariates $X$ and parameterized by a finite-dimensional threshold
parameter $\theta$. Consider a family of distributions $\mathcal{P}$ such that for all $P \in \mathcal{P}$, 
Assumption~\ref{cond:known ps Donsker} holds. 
Moreover, for every $P \in \mathcal{P}$, both the conditional benefit $\E_{P}[\tau \mid X]$ and the conditional cost $\E_{P}[C \mid X]$ are strictly positive. 
In particular, there exists  $\underline{c} > 0$ such that 
\[
\E_{P}[C \mid X] \geq \underline{c} \quad \text{for all } P \in \mathcal{P},
\]
so that Assumption~\ref{assu:gain when exactly binding} is violated for every $P \in \mathcal{P}$.
  Then the modified sample-analog rule  $\widehat{g}_{\alpha_{n}}$, as defined in~\eqref{eq:EWM-size}, 
is  asymptotically welfare-efficient and asymptotically feasible uniformly under
$\mathcal{P}$.
\end{cor}

Note that the above results do not contradict the impossibility result in Section~\ref{subsec:Mistakes-from-any}, which shows lack of uniformity over a broader class of distributions, including those satisfying Assumptions \ref{assu:exactly binding optimal} and \ref{assu:gain when exactly binding}. Specifically, Corollary~\ref{cor:optimal mistake control} establishes uniformity within a subclass of distributions that violate  Assumption~\ref{assu:exactly binding optimal}, while Corollary~\ref{cor:positive cost} establishes uniformity within a subclass of distributions that violate  Assumption~\ref{assu:gain when exactly binding}.
\section{\label{sec:Trade-off}Trade-off rule}
Up to this point, the discussion has assumed that any budget violation is undesirable. If policymakers are willing to borrow, potentially at some penalty, to exceed the budget constraint in order to achieve higher welfare, then alternative rules might be preferable.    Section \ref{subsec:Form-trade-off}
first formalizes this setting as a trade-off problem. Section
\ref{subsec:Estimator-tradeoff} derives a statistical rule that implements
such trade-off in the sample, and   is shown to be uniformly asymptotically
welfare-efficient. 

\subsection{Form of the trade-off\label{subsec:Form-trade-off}}

Exceeding the budget can have negative economic consequences, such as increased borrowing costs or reduced funding for other programs. However, it may not be severe enough to entirely outweigh the perceived welfare gains. Consider a new objective function in which the policymaker, while operating within the budget constraint, seeks to allocate resources efficiently to maximize welfare, but once the budget is exceeded, must trade off the welfare gains from additional spending against the penalty associated with the overrun. Let $\overline{\lambda}>0$ denote this penalty. Using the notation $(x)_{+}$ to represent the positive part of $x\in\mathbb{R}$, the objective function can be written as\footnote{This objective function can also be motivated as policymakers might be willing to trade off violations of the constraint against gains  in welfare only to a certain extent, bounding  the marginal gain of relaxing the constraint $\lambda\in[0,\overline{\lambda}]:$  \begin{equation*} \max_{g\in\mathcal{G}}\min_{\lambda\in[0,\overline{\lambda}]}\ W_{}(g;P)-\frac{\lambda}{r}\cdot(B(g;P)-k).\label{eq:dual-upper bound}
 \end{equation*}
 }
\begin{equation}
\max_{g\in\mathcal{G}}V(g;P)\text{ where }V(g;P)\coloneqq W_{}(g;P)-\frac{\overline{\lambda}}{r}\cdot(B(g;P)-k)_{+}.\label{eq:hinge objective}
\end{equation}
Here $r>0$ denotes the rate at which monetary units are converted into welfare units, since welfare is often not measured in monetary terms whereas the budget is. 

This objective function~\eqref{eq:hinge objective} is a non-smooth but piecewise linear optimization
problem. Denote its solution to be:
\begin{equation}
\widetilde{g}_{P}\in\arg\max_{g\in\mathcal{G}}\ V(g;P),\label{eq:hinge population}
\end{equation} 
and note that the solution $\widetilde{g}_{P}$ can relax the budget constraint and achieves weakly higher
welfare than the constrained optimal policy $g_{P}^{\ast}$ for
any data distribution $P$. The next lemma formalizes this observation.
\begin{lem}
\label{lem:hinge higher value}For any $\overline{\lambda}\in[0,\infty)$,
the solution to the trade-off problem (\ref{eq:hinge population})
achieves weakly higher welfare than the constrained optimum: $W(\widetilde{g}_{P};P)\geq V(\widetilde{g}_{P};P) \geq W(g_{P}^{\ast};P)$
for any distribution $P$. Moreover, the violation to the budget constraint
is upper bounded by $\frac{r\cdot (W(\widetilde{g}_{P};P)-W(g_{P}^{\ast};P))}{\overline{\lambda}}$.
\end{lem}

\subsection{Welfare efficiency of the trade-off rule\label{subsec:Estimator-tradeoff}}

Given a new objective function (\ref{eq:hinge population}) that trades
off the gain and the cost from violating the constraint, the goal
is to derive a statistical rule that is likely to select eligibility
policies that maximize the new objective function $V(g;P)$. Consider the trade-off
statistical rule defined as 
\begin{equation}
\widehat{g}_{\text{tradeoff}}\in\arg\max_{g\in\mathcal{G}}\ \widehat{V}_{n}(g)\text{ where }\widehat{V}_{n}(g)\coloneqq \widehat{W}_{n}(g)-\frac{\overline{\lambda}}{r}\cdot(\widehat{B}_{n}(g)-k)_{+},\label{eq:hinge}
\end{equation}
where the subscript ``tradeoff'' highlights that this statistical
rule is able to relax the constraint by trading off the gain and the
cost from violating the constraint. Since $\widehat{g}_{\text{tradeoff}}$ solves a sample-analog version of (\ref{eq:hinge population}),  I show it consistently achieves the maximal value of $V(g;P)$ under weak conditions in Lemma~\ref{lem:hinge consistency}. Theorem~\ref{thm:hinge tradeoff} further verifies that it consistently achieves welfare that is weakly higher than $W(g_{P}^{\ast};P)$. I leave to future research to verify whether the trade-off rule is minimax rate optimal.
\begin{lem}
    \label{lem:hinge consistency}
    Suppose Assumption \ref{cond:known ps Donsker}
holds for the class of data distributions $\mathcal{P}$. Then the
trade-off rule $\widehat{g}_{\text{tradeoff}}$ defined in (\ref{eq:hinge})
consistently achieves the maximized value of $V(g;P)$, namely $V(\widetilde{g}_{P};P).$
\end{lem}
\begin{thm}
\label{thm:hinge tradeoff}For the class of data distributions $\mathcal{P}$ satisfying Assumption \ref{cond:known ps Donsker}, the
trade-off rule $\widehat{g}_{\text{tradeoff}}$  
consistently achieves welfare that is weakly higher than $W(g_{P}^{\ast};P)$ for all $P\in\mathcal P$ and is therefore uniformly asymptotically welfare-efficient under $\mathcal{P}$.
Moreover, with probability approaching one, the violation to the budget
constraint is upper bounded by $\frac{r\cdot (W(\widehat{g}_{\text{tradeoff}};P)-W(g_{P}^{\ast};P))}{\overline{\lambda}}$
uniformly over $\mathcal{P}$.
\end{thm}
To gain intuition for the above results, note that Lemma~\ref{lem:hinge consistency} shows the trade-off
rule $\widehat{g}_{\text{tradeoff}}$ uniformly consistently achieves $V(\widetilde{g}_{P};P),$ which by Lemma \ref{lem:hinge higher value} is weakly higher than the welfare achieved by 
the constrained optimal policy $g_{P}^{\ast}$ for any data distribution
$P$.  Therefore, the trade-off rule $\widehat{g}_{\text{tradeoff}}$
is asymptotically welfare-efficient uniformly over $\mathcal{P}$. At the same time, larger   $\bar{\lambda}$ implies smaller violation to the budget constraint, relative to the welfare gain $W(\widehat{g}_{\text{tradeoff}};P)-W(g_{P}^{\ast};P)$.

\subsection{Medicaid expansion: empirical illustration\label{sec:Empirical-example}}

Example \ref{exa:Welfare-program-eligibility}
of Section \ref{sec:Motivations-and-setup} explains a more flexible Medicaid expansion policy that would allow the income thresholds to vary
with the number of children. In this example, while the budget constraint is set equal to the cost of the current policy, policymakers would be willing to exceed the budget constraint, potentially at some penalty, in order to achieve higher welfare. This subsection  uses this example, together with data from  the Oregon Medicaid Health Insurance Experiment (OHIE), to illustrate the trade-off rule and compare it with the sample-analog rule  and its modification. %

\subsubsection{Data\label{subsec:Data}}

I use the experimental data from the OHIE, where Medicaid eligibility
($D_{i}$) was randomized in 2007 among Oregon residents who were
low-income adults, but previously ineligible for Medicaid, and who
expressed interest in participating in the experiment. \citet{finkelstein_oregon_2012}
include a detailed description of the experiment and an assessment
of the average effects of Medicaid on health and health care utilization.
I include a cursory explanation here for completeness.%

{} The original OHIE sample consists of 74,922 individuals (representing
66,385 households). Of these, 26,423 individuals responded to the
initial mail survey, which collects information on income as percentage
of the federal poverty level and number of children, which are the
characteristics of interest for targeting ($X_{i}$).\footnote{More accurately, I follow \citet{sacarny_out_2020} to approximate
number of children by the number of family members under age 19 living
in house as reported on the initial mail survey. I exclude individuals
who did not respond to the initial survey from my sample, which differs
from the sample analyzed in \citet{finkelstein_oregon_2012} as I
focus on individuals who responded both to the initial and the main
surveys from the OHIE. Due to this difference, the expansion policies
selected using my sample do not directly carry their properties to
the population underlying the original OHIE sample, as the distributions
of $X$ differ. } After one year, the main survey collects data related to health ($Y_{i}$),
health care utilization ($H_{i}$) and actual enrollment in Medicaid
($M_{i}$), which allows me to construct estimates for the benefit
and cost of Medicaid eligibility $(\tau,C)$. Therefore I further
exclude individuals who did not respond to the main survey from my
sample. %

For health ($Y_{i}$), I  follow the binary measurement in \citet{finkelstein_oregon_2012}
based on self-reported health, where an answer of ``poor/fair''
is coded as $Y_{i}=0$ and ``excellent/very good/good'' is coded
as $Y_{i}=1$. For health care utilization ($H_{i}$), the study collected
measures of utilization of prescription drugs, outpatient visits,
ER visits, and inpatient hospital visits. \citet{finkelstein_oregon_2012}
annualize these utilization measures to turn these into spending estimates,
weighting each type by its average cost (expenditures inflated with the CPI-U to 2007 dollars) among low-income
publicly insured non-elderly adults in the Medical Expenditure Survey
(MEPS). Note that health and health care utilization are not measured
at the same scale, which requires rescaling when I  consider the trade-off
between the two. I  address this issue in Section \ref{subsec:Estimating-a-more}.
Lastly, since the enrollment in Medicaid still requires an application,
not everyone eligible in the OHIE eventually enrolled in Medicaid,
which implies $M_{i}\leq D_{i}$.

Given the setup of the OHIE, Medicaid eligibility ($D_{i}$) is random
conditional on household size (number of adults in the household)
entered on the lottery sign-up form and survey wave. While the original
experimental setup would ensure randomization given household size,
the OHIE had to adjust randomization for later waves of survey respondents
(see the Appendix of \citet{finkelstein_oregon_2012} for more details).
Denote the confounders (household size and survey wave) with $V_{i}$,
and define the propensity score as $p(V_{i})=\Pr\{D_{i}=1\mid V_{i}\}$.
If the propensity score is known, then the construction of the estimates
follows directly from the formula (\ref{eq:IPW}). However, the adjustment
for later survey waves means I need to estimate the propensity score,
and I adapt the formula (\ref{eq:IPW}) following \citet{athey_policy_2021}
to account for the estimated propensity score.

Specifically, define the conditional expectation function (CEF) of
a random variable $U_{i}$ as $\gamma^{U}=\E[U_{i}\mid V_{i},D_{i}${]}.
Since $V_{i}$ in my case is discrete, I use a fully saturated model
to estimate the propensity score $\widehat{p}(V_{i})$ and the CEF
$\widehat{\gamma}^{U}(V_{i},D_{i})$. I then form the estimated Horvitz-Thompson
weight with the estimated propensity score as $\widehat{\alpha}(V_{i},D_{i})=\frac{D_{i}}{\widehat{p}(V_{i})}-\frac{1-D_{i}}{1-\widehat{p}(V_{i})}$.
For health benefit due to Medicaid eligibility, define the estimate
$\tau_{i}^{\ast}=\widehat{\gamma}^{Y}(V_{i},1)-\widehat{\gamma}^{Y}(V_{i},0)+\widehat{\alpha}(V_{i},D_{i})\cdot\left(Y_{i}-\widehat{\gamma}^{Y}(V_{i},D_{i})\right)$. For the cost due to Medicaid eligibility, define the estimate $C_{i}^{\ast}=\widehat{\gamma}^{Z}(V_{i},1)+\frac{D_{i}}{\widehat{p}(W_{i})}\cdot\left(Z_{i}-\widehat{\gamma}^{Z}(V_{i},D_{i})\right)$ where $Z_{i}=M_{i} \cdot H_{i}$. Since an eligible individual only incurs cost to Medicaid if enrolled, I  need to account for imperfect take-up in forming $C_{i}^{\ast}$.

\begin{table}[H]
\begin{centering}
\caption{\label{tab:Sample-characteristics}Summary statistics of the OHIE
sample by number of children}
\begin{tabular}{cccc}
\toprule 
Number of children & Sample size & Sample mean of $\tau_{i}^{\ast}$ & Sample mean of $C_{i}^{\ast}$\tabularnewline
\midrule
\midrule
0 & 5,758 & 3.16\% & \$1,974\\
  &  & (0.01) & (110)\\
\midrule
1 & 1,736 & 10.34\% & \$1,615\\
  &  & (0.02) & (195)\\
\midrule
$\geq2$ & 2,641 & 1.55\% & \$1,451\\
  &  & (0.02) & (129)\\
\midrule
Full sample & 10,135 & 3.97\% & \$1,776\\
  &  & (0.01) & (79)\\
\bottomrule
\end{tabular}
\par\end{centering}
\emph{Notes:} This table presents summary statistics on the sample
of individuals who responded to both the initial and the main surveys
from the Oregon Health Insurance Experiment (the OHIE sample). The
first three rows represent individuals living with different number
of children (family members under age 19), and the last row is the
aggregate. The estimate for benefit $\tau_{i}^{\ast}$ is an estimate
for the increase in the probability of an individual reporting ``excellent/very
good/good'' on self-reported health (as opposed to ``poor/fair'')
after receiving Medicaid eligibility. The estimate for cost $C_{i}^{\ast}$
is an estimate for individual's health care expenditure that needs
to be reimbursed by Medicaid. Standard errors are shown in parentheses below.
\end{table}
Table~\ref{tab:Sample-characteristics} presents the summary statistics.  While Appendix
\ref{subsec:Results-with-estimated} argues the estimation errors in $\tau_{i}^{\ast}$
and $C_{i}^{\ast}$ are asymptotically negligible, in finite samples, the cost estimates are highly variable, resulting in noisy estimate $\widehat{B}_n(g)$. 

To formalize the budget constraint requiring that the average cost of any proposed policy does not exceed that of the status quo 2014 Medicaid expansion, I calibrate the per capita cost under the 2014 policy. Following  \cite{finkelstein_oregon_2012}, who cite \citet{wallace2008effective}, Medicaid spending among individuals comparable to the Oregon Health Insurance Experiment (OHIE) participants was approximately \$3,000 per enrollee in Oregon in 2004, which corresponds to about \$3,300 in 2007 dollars.  Under the 2014 policy,  I adjust for imperfect take-up by defining the per capita cost as \begin{equation}
    k = \E_{P}[M(1)g_{2014}(X)]\cdot\$3,300 = \$1,377 \label{eq:def k}
\end{equation} where $M(1)$ denotes enrollment status if offered Medicaid eligibility and $g_{2014}(x)=\mathbf{1}\{\text{income}\leq 138\%\}$ represents the status quo 2014 expansion policy of providing eligibility to all adults with income up to 138\%. The enrollment rate under the status quo 2014 policy $\E_{P}[M(1)g_{2014}(X)]$ is based on  point estimate from the OHIE.\footnote{The theoretical results in this paper are developed for a fixed $k$. Addressing the estimation error in 
$k$ is left to future work.}

\subsubsection{Budget-constrained Medicaid expansion\label{subsec:Estimating-a-more}}

Figure \ref{fig:Conditional-eligibility-income} summarizes the selected
expansion policies, which are income thresholds specific to the number
of children. The sample-analog rule $\widehat{g}_{\text{sample}}$ chooses to restrict Medicaid eligibility, especially lowering the income threshold for childless individuals far below the current level, and the estimated welfare is 3.79\% increase in reporting good subjective health. The budget estimate for the selected policy is \$1,311,  slightly below the threshold of $k=\$1,377$, as $\widehat{g}_{\text{sample}}$ imposes the  sample-analog version of the budget constraint. However, due to the large variation in the cost estimates as illustrated in Table~\ref{tab:Sample-characteristics}, meeting the sample budget constraint still involves uncertainty about whether the selected policy meets the budget constraint in the population as argued in Proposition~\ref{prop: EWM no ptwise mistake}. After taking into account of estimation uncertainty, the modified sample-analog rule $\widehat{g}_{\alpha=5\%}$ is much more conservative than the sample-analog rule $\widehat{g}_{\text{sample}}$. 
 The budget estimate of the selected policy \$1,174, with a standard error of 66, making it statistically significantly below the budget constraint at the conventional 5\% level. The welfare estimate of the selected policy is lower at 3.43\%. One can reduce the conservativeness by increasing $\alpha$ and therefore lowering the statistical guarantee that the selected policy meets the budget constraint in the population. I examine the results under higher values of $\alpha$ in Appendix~\ref{sec:robustness}.

To construct the trade-off rule $\widehat{g}_{\text{tradeoff}}$ as
proposed in Section \ref{subsec:Form-trade-off}, I need to specify both the penalty parameter, $\overline{\lambda}$ and the conversion rate, $r$. For illustration, I assume that exceeding the budget incurs a full repayment of the overrun, implying $\overline{\lambda}=1.$
In my empirical illustration, the budget constraint is in terms of
monetary value. The objective function, however, is measured based
on self-reported health, which does not directly translate to a monetary
value. \citet{finkelstein_value_2019} converts self-reported
health into value of a statistical life year (VSLY) based on existing
estimates. Specifically, a conservative measure for the increase in
quality-adjusted life year (QALY) when self-reported health increases
from ``poor/fair'' to ``excellent/very good/good'' is roughly
0.6. The \textquotedblleft consensus\textquotedblright{} estimate
for the VSLY for one unit of QALY from \citet{cutler_your_2004} is
\$100,000 for the general US population. I therefore follow \citet{finkelstein_value_2019} and set $r = (0.6\cdot\$100,000) = \$60,000.$
The trade-off objective function is therefore
\begin{equation}
\max_{g\in\mathcal{G}}V(g;P)\text{ where }V(g;P)\coloneqq W_{}(g;P)-\frac{1}{\$60,000}\cdot(B(g;P)-k)_{+}.\label{eq:OHIE tradeoff}
\end{equation}

The trade-off rule $\widehat{g}_{\text{tradeoff}}$ based on~\eqref{eq:OHIE tradeoff} chooses
to assign Medicaid eligibility to more individuals, and to raise the
income thresholds above the current level for those with children. Therefore the welfare estimate for the selected policy is higher at 3.86\% and
budget estimate for the selected policy is \$1,402, slightly above the budget constraint. The higher level occurs
because on average the benefit estimates are positive, 
and the trade-off rule finds that the additional health benefit from violating
the budget constraint exceeds the cost of doing so, especially those with children. While a traditional cost-benefit analysis would also recommend prioritizing those with children based on Table~\ref{tab:Sample-characteristics},  the trade-off rule offers a more interpretable recommendation by directly selecting income thresholds. 

 However,  setting $\overline{\lambda}=1$ assumes that the penalty is limited to repaying the budget overrun, potentially understating the true penalty for policymakers, as it does not account for additional economic costs associated with exceeding the budget, such as reduced funding for other programs or lower quality of Medicaid services. Accordingly, in Appendix~\ref{sec:robustness}, I explore how the results change when $\overline{\lambda}$ is increased, which is equivalent to making budget overruns more costly to the policymaker. For example, with $\overline{\lambda}=1.68$, the welfare and budget estimates for the selected policy  essentially drop to those of $\widehat{g}_{\text{sample}}$.%
 
\begin{figure}[H]
\begin{centering}
\caption{\label{fig:Conditional-eligibility-income}More flexible Medicaid
expansion policies}
\subfloat[ policy selected by the sample-analog rule $\widehat{g}_{\text{sample}}$]{\begin{centering}
\par\end{centering}
\centering{}\includegraphics[scale=0.25]{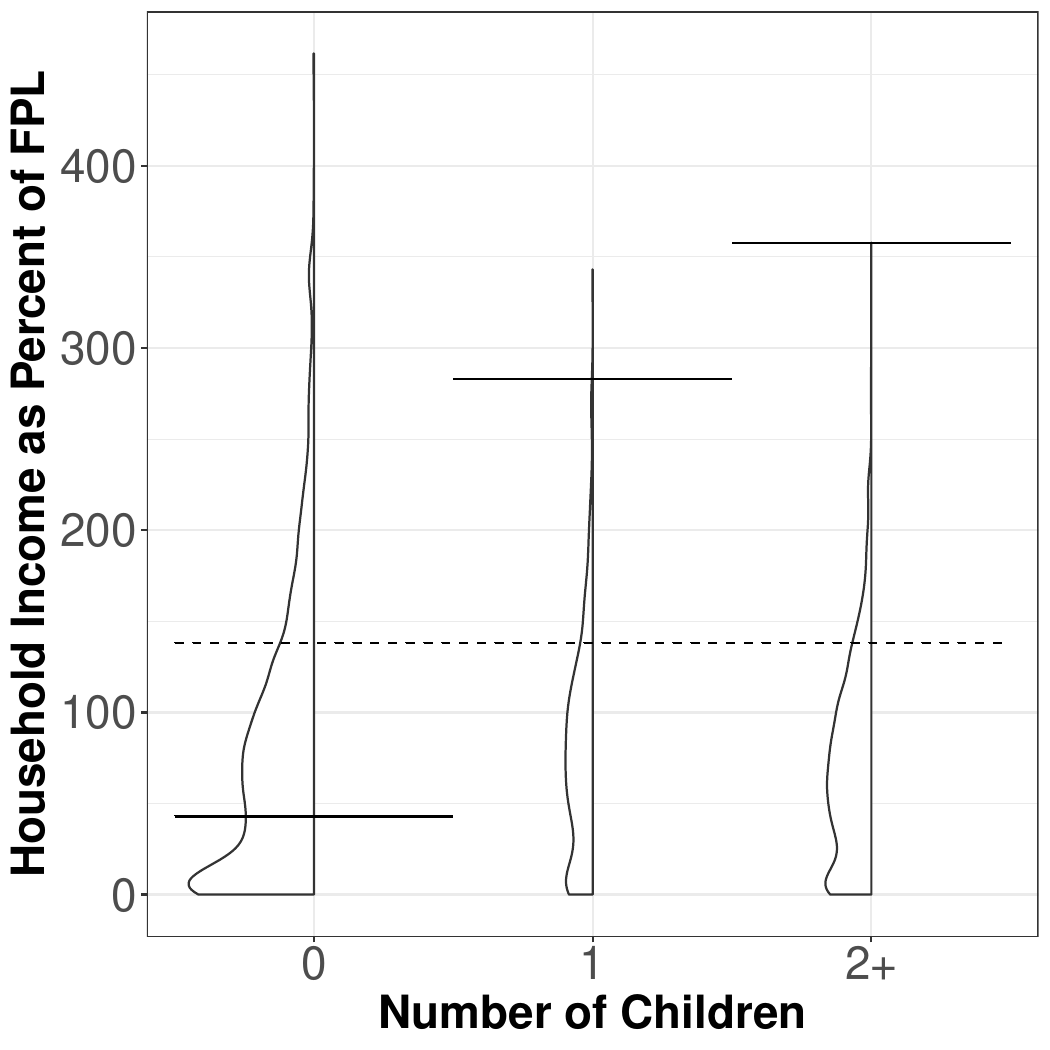}}
\hfill{}\subfloat[ policy selected by the modified rule $\widehat{g}_{\alpha=5\%}$]{
\centering{}\includegraphics[scale=0.25]{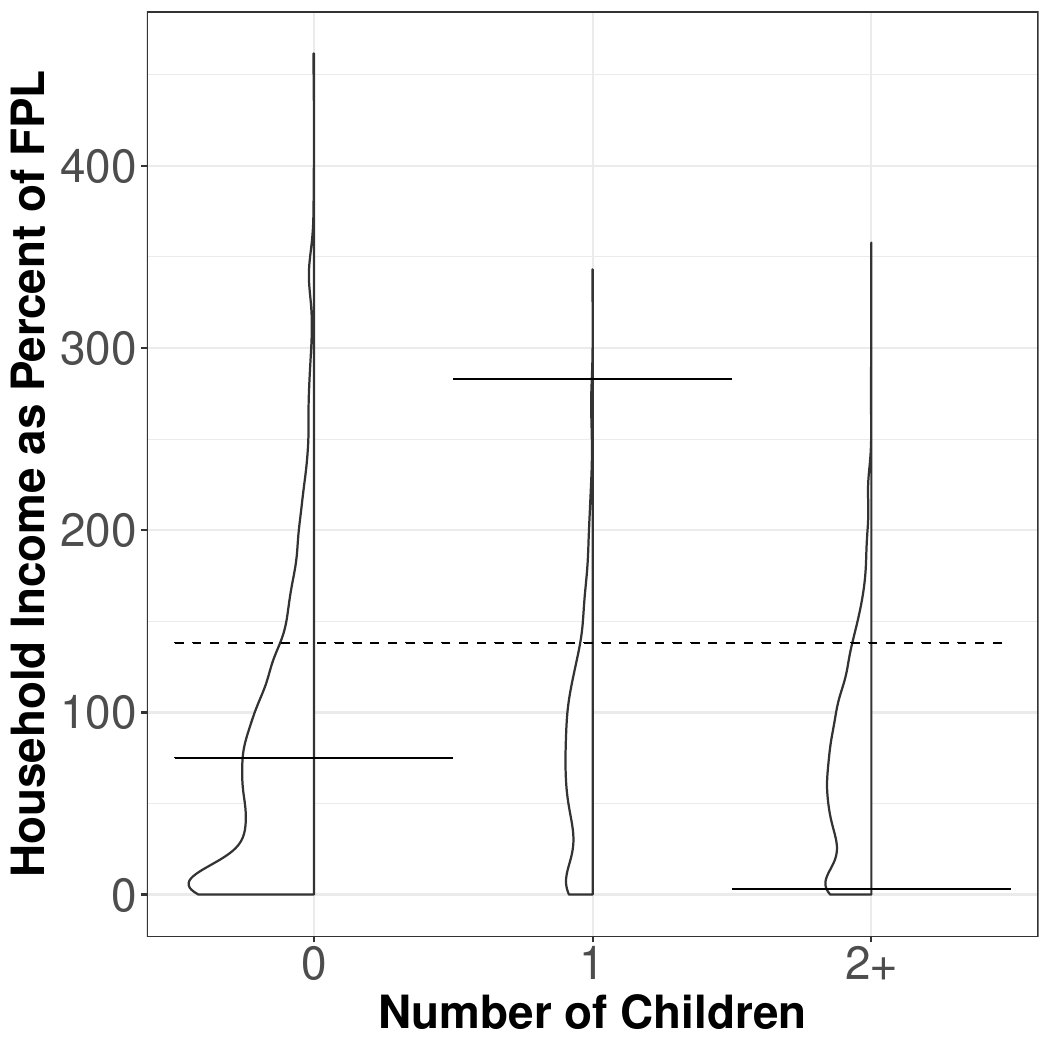}}
\hfill{}\subfloat[ policy selected by the trade-off rule $\widehat{g}_{\text{tradeoff}}$]{\begin{centering}
\par\end{centering}
\centering{}\includegraphics[scale=0.25]{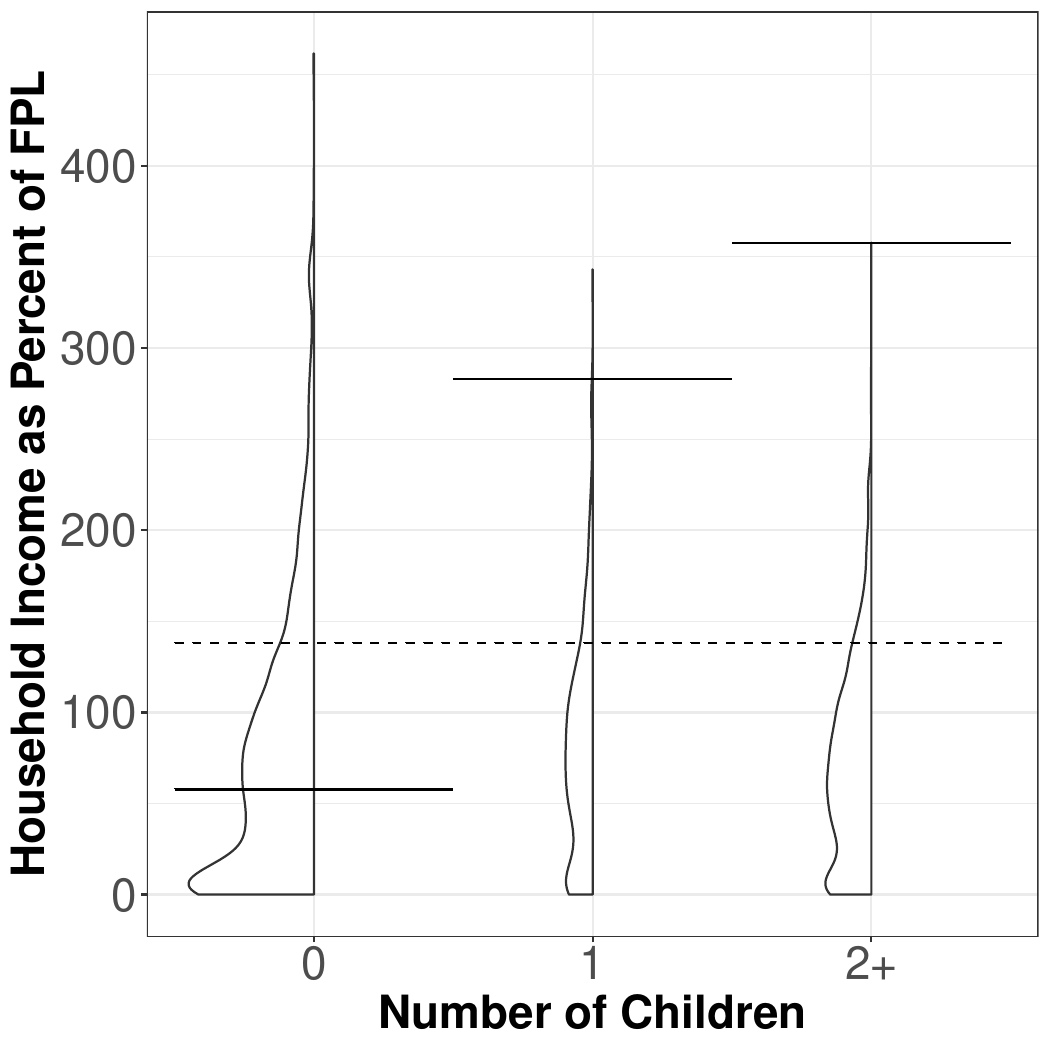}}
\par\end{centering}
\raggedright{}\emph{Notes: }This figure plots the more flexible Medicaid
expansion policies selected by statistical rules  
based on results from the OHIE. The horizontal dashed line marks the
income thresholds under the current expansion policy, which is
138\% regardless of the number of children in a household. The horizontal
solid lines mark the more flexible policy selected by various statistical
rules, i.e. income thresholds that can vary with number of children.
For each number of children, I also plot the underlying income distribution
to visualize individuals below the thresholds. Panel (a) plots the
 policy selected by the sample-analog rule $\widehat{g}_{\text{sample}}$. Panel (b) plots the
 policy selected by the modified sample-analog rule $\widehat{g}_{\alpha=5\%}$.
Panel (c) plots the policy selected by the trade-off rule $\widehat{g}_{\text{tradeoff}}$.
\end{figure}

In this example, the policies chosen by different rules differ substantially.  A natural question is how policymakers should choose the
rules.   In this Medicaid example, it is reasonable to assume that policymakers can borrow, potentially at some penalty, to exceed the budget constraint in order to achieve higher welfare. Therefore the trade-off rule $\widehat{g}_{\text{tradeoff}}$ is theoretically attractive.  It selects policies that achieve at least the maximum feasible
welfare, while accounting for the penalty $\overline{\lambda}$ from violating the budget constraint.  Because the amount of budget violations depend on the choice of $\overline{\lambda}$ as shown in Theorem~\ref{thm:hinge tradeoff},  one can
assess this trade-off for their particular settings by varying $\overline{\lambda}$. However, in other context where policymakers are financially conservative and even minor violation to the budget constraint is unacceptable, then  the modified sample-analog rule $\widehat{g}_{\alpha}$, with a pre-specified significance level $\alpha$, is theoretically attractive. It offers a statistical guarantee that the budget constraint will be met in the population with high confidence. However, the selected eligibility policy may appear very conservative, as illustrated in the Medicaid expansion example above.

\section{Monte Carlo Simulations\label{sec:Simulation-study}}

To ensure the practical relevance of the simulation, I calibrate to the distribution of the data from the OHIE, which is also used for empirical illustration in Section \ref{sec:Empirical-example}, based on Example \ref{exa:Welfare-program-eligibility}
of Section \ref{sec:Motivations-and-setup}. 

For the purpose of this simulation
study, the OHIE represents the population $P$, and I 
take the estimates $(\tau_{i}^{\ast},C_{i}^{\ast})$ constructed in Section \ref{subsec:Data} as the true benefit and cost $(\tau,C)$. 
Under this simulation design, I can solve for the constrained
optimal policy as 
\[
g_{P}^{\ast}\in\arg\max_{g\in\mathcal{\widetilde{G}},\ B(g;P)\leq k}W_{}(g;P)
\]
where $\left(W_{}(g;P),B(g;P)\right)$ are the sample analogs in the
OHIE sample. The policy class $\mathcal{\widetilde{G}}$ includes income thresholds
that can vary with the number of children as described in Equation~\eqref{eq:coarse grid} of Appendix~\ref{appx:computation}, which is slightly coarsened than the class used in Section~\ref{subsec:Estimating-a-more}. The maximum feasible welfare is given by $W_{}(g_{P}^{\ast};P)=3.76\%$, an increase
of 3.76\% in reporting good subjective health. The cost associated
with the constrained optimal policy is $B(g_{P}^{\ast};P)=\$1,340$,
slightly below the constraint $k=\$1,377$ as defined in Equation~\eqref{eq:def k}.

\subsection{Simulation results}

Table \ref{tab:Asymptotic-properties-of} compares the performance
of various statistical rules $\widehat{g}$ through 500 Monte Carlo
iterations. At each iteration, I randomly draw observations from the
OHIE sample to form a random sample. I simulate with the same sample
size as the original sample to hold the amount of sampling uncertainty
constant. Given the random sample, I collect eligibility policies
chosen by each of the following statistical rules:
\begin{itemize}
\item sample-analog rule $\widehat{g}_{\text{sample}}$,
\item modified sample-analog rule $\gmistake$   described in Theorem \ref{thm: uniform size} with $\alpha$ set to 5\%,
\item trade-off rule $\widehat{g}_{\text{tradeoff}}$ with  
$\bar{\lambda}=1$ and $r=\$ 60,000$, same as the empirical illustration as explained in Equation~\eqref{eq:OHIE tradeoff}.
\end{itemize}
I evaluate the welfare function and the budget function $\left(W_{}(g;P),B(g;P)\right)$
for a given policy in the original OHIE sample. Averages over 500
iterations provide simulation evidence on the  properties
of the above statistical rules, as shown in Table \ref{tab:Asymptotic-properties-of}.
\begin{table}[H]
\begin{centering}
\caption{\label{tab:Asymptotic-properties-of}Simulation results}
\begin{tabular}{cccc}
\toprule 
Statistical rule & sample-analog & modified & trade-off\tabularnewline
 & $\widehat{g}_{\text{sample}}$ & $\widehat{g}_{\alpha=5\%}$ & $\widehat{g}_{\text{tradeoff}}$\tabularnewline
\midrule
\midrule 
Prob. of selecting infeasible policies & 10.2\% & 0.2\% & 45.2\% \\
\midrule
Prob. of selecting suboptimal policies & 94.0\% & 99.4\% & 65.6\% \\
\midrule
Average welfare loss & 0.06 & 0.17 & 0.01 \\
\midrule
Average budget overrun & -\$70 & -\$278 & \$62 \\
\bottomrule
\end{tabular}
\par\end{centering}
\emph{Notes:} This table reports  properties of statistical
rules $\widehat{g}$, as averaged over 500 simulations. Row 1 reports
the probability that the rule selects an eligibility policy that
violates the budget constraint, i.e. $\Pr_{P^{n}}\{B(\widehat{g};P)>0\}$.
Row 2 reports the probability that the rule achieves strictly less
welfare than the constrained optimal policy  $g_{P}^{\ast}$, i.e.
$\Pr_{P^{n}}\{W_{}(\widehat{g};P)<W_{}(g_{P}^{\ast};P)\}$. Row 3 reports
the average welfare loss of the rule relative to the maximum feasible
welfare, i.e. $\frac{E_{P^{n}}\left[W(g_{P}^{\ast};P)-W(\widehat{g};P)\right]}{W(g_{P}^{\ast};P)}$.
Row 4 reports the average budget overrun of the policies selected by the rule,
i.e. $E_{P^{n}}\left[B(\widehat{g};P)\right]-k$. 
\end{table}

Row 1 of Table \ref{tab:Asymptotic-properties-of} illustrates that
it is possible for all three statistical rule $\widehat{g}$ to select
infeasible policies. A lower probability of selecting infeasible policies
suggests the rule is closer to achieving asymptotic feasibility. 
In the distribution calibrated to the OHIE sample, the original sample-analog rule $\widehat{g}_{\text{sample}}$ might not be asymptotically feasible as it can select infeasible eligibility policies in 10.2\% of the draws. 
In contrast, Theorem \ref{thm: uniform size} guarantees
that a simple modification $\widehat{g}_{\alpha=5\%}$
selects infeasible eligibility policies in less than 5\% of the draws,
regardless of the distribution. Simulation confirms such guarantee
as the mistakes only happen 0.2\% of the time. 

Row 2 of Table \ref{tab:Asymptotic-properties-of} illustrates that
it is possible for all three statistical rule $\widehat{g}$ to achieve
weakly higher welfare than the constrained optimal policy $g_{P}^{\ast}$.
This can happen when  $\widehat{g}$ selects an infeasible policy.
A lower probability of selecting suboptimal policies suggests the
rule is closer to achieving asymptotic welfare-efficiency. Theorem
\ref{thm:hinge tradeoff} implies that the trade-off rule $\widehat{g}_{\text{tradeoff}}$
is uniformly asymptotically welfare efficient while there is no such
guarantee for the sample-analog rule $\widehat{g}_{\text{sample}}$. 

In the distribution calibrated to the OHIE, the trade-off rule $\widehat{g}_{\text{tradeoff}}$
on average achieves higher welfare than the sample-analog rule $\widehat{g}_{\text{sample}}$.
As shown in row 3 of Table \ref{tab:Asymptotic-properties-of}, the
welfare loss of $\widehat{g}_{\text{tradeoff}}$ is 1\% of the maximum
feasible welfare $W_{}(g_{P}^{\ast};P)$, compared to 6\% for $\widehat{g}_{\text{sample}}$.
However, its improvement can be at the cost of violating the budget
constraint more often than $\widehat{g}_{\text{sample}}$, at a rate
of 45.2\%. Though as shown in row 4 of Table \ref{tab:Asymptotic-properties-of}, on average the violation is limited, which confirms Theorem
\ref{thm:hinge tradeoff}.

\section{Conclusion}

In this paper, I focus on properties of statistical rules when the
cost of implementing any given   policy needs to be estimated.
The existing EWM rule selects an eligibility policy that maximizes
a sample analog of the social welfare function, and only accounts
for constraints that can be verified with certainty in the population.
However, in some cases, the cost of providing eligibility to any given individual
might be unknown ex-ante due to imperfect take-up and heterogeneity.
Therefore, in addition to asymptotic welfare-efficiency that has been
studied by the EWM literature, I introduce a new desirable property
of statistical rules in the setting of unknown cost, namely asymptotic
feasibility, which requires the selected policy to satisfy a budget constraint when the sample size is large enough. Unlike the setting of known cost, I prove an impossibility result that no statistical rule
can be uniformly asymptotically welfare-efficient and feasible. The direct extension
to the existing EWM approach is no longer asymptotically welfare efficient
nor asymptotically feasible for certain real-world relevant data distributions. As an alternative, I propose the trade-off rule that guarantees asymptotic welfare efficiency while ensuring any budget violations are bounded above by welfare gains. I illustrate the theoretical results using experimental data from the OHIE. 
A promising avenue
for future research is to verify whether the  trade-off rule maintains minimax rate optimality in the constrained setting, just as how the EWM rule is in the unconstrained setting.

\begin{singlespace}
\bibliographystyle{aer}
\bibliography{jacknife}

\end{singlespace}

\newpage
\appendix

\section{Proofs of theorems}

\begin{proof}
\textbf{Proof of Theorem \ref{thm: no uniform estimator}.} The population
is a probability space $(\Omega,\mathcal{A},P)$, which induces the
sampling distribution $P^{n}$ that governs the observed sample. A
statistical rule $\widehat{g}$ is a mapping $\widehat{g}(\cdot):\Omega\rightarrow\mathcal{G}$
that selects a policy from the policy class $\mathcal{G}$ based on
the observed sample. Note that the selected policy $\widehat{g}(\omega)$
is still deterministic because the policy class $\mathcal{G}$ is
restricted to be deterministic policies. When no confusion arises, I drop the reference
to event $\omega$ for notational simplicity.

Suppose $\widehat{g}$ is asymptotically welfare-efficient and asymptotically
feasible under $P_{0}$. We want to prove there is non-vanishing chance
that $\widehat{g}$ selects policies that are infeasible some other
distributions:
\begin{align*}
\limsup_{n\rightarrow\infty}\sup_{P\in\mathcal{P}}\Pr_{P^{n}}\left\{ \omega:B(\widehat{g}(\omega);P)>k\right\}  & >0.\label{eq:ptwise mistake ph-1}
\end{align*}

Firstly, asymptotical welfare-efficiency of $\widehat{g}$ under $P_{0}$ implies for any $\epsilon>0$
we have
\begin{equation*}
\limsup_{n\rightarrow\infty}\Pr_{P_{0}^{n}}\left\{ \omega:W(g_{p_{0}}^{\ast}P_{0})-W(\widehat{g}(\omega);P_{0})>\epsilon\right\} =0.\label{eq:ptwise no regret p0 estimator}
\end{equation*}
Asymptotical feasibility under $P_{0}$ implies $\Pr_{P_{0}^{n}}\{\omega:B(\widehat{g}(\omega);P_{0})\leq k\}\rightarrow1.$ 

Consider the event $\omega'$ where $\left|W(\widehat{g}(\omega');P)-W(g_{P}^{\ast};P)\right|<\epsilon$
and $\widehat{g}(\omega')$ is feasible. Asymptotic welfare-efficiency
and asymptotic feasibility  imply $\Pr_{P_{0}^{n}}\left\{ \omega'\right\} \rightarrow1.$
To see this, note the probability of such event has an asymptotic
lower bound of one 
\begin{align*}
 & \Pr_{P_{0}^{n}}\left\{ \omega':W(g_{p_{0}}^{\ast}P_{0})-W(\widehat{g}(\omega');P_{0})\leq\epsilon\text{ and }B(\widehat{g}(\omega');P_{0})\leq k\right\} \\
\geq & \Pr_{P_{0}^{n}}\left\{ \omega:W(g_{p_{0}}^{\ast}P_{0})-W(\widehat{g}(\omega);P_{0})\leq\epsilon\right\} +\Pr_{P_{0}^{n}}\left\{ \omega:B(\widehat{g}(\omega);P_{0})\leq k\right\} -1
\end{align*}
where the first two terms converge to one as $n\rightarrow\infty$
respectively under asymptotic welfare-efficiency and asymptotic feasibility.

Now consider $P_0$ that satisfies Assumptions \ref{assu:exactly binding optimal} and \ref{assu:gain when exactly binding}. For $\epsilon$ in Assumption \ref{assu:gain when exactly binding},
we have $B(\widehat{g};P_{0})=B(g_{P_{0}}^{\ast};P_{0})=k$ under
the event $\omega'$, since the constraint is exactly satisfied at
$g_{P_{0}}^{\ast}$ under Assumption \ref{assu:exactly binding optimal}.
By Law of Total Probability, we have 
\[
\Pr_{P_{0}^{n}}\left\{ \omega:B(\widehat{g}(\omega);P_{0})=k\right\} \geq \Pr_{P_{0}^{n}}\left\{ \omega'\right\} \cdot \Pr_{P_{0}^{n}}\left\{ B(\widehat{g};P_{0})=k\mid\omega'\right\} 
\]
Then the above argument shows $\Pr_{P_{0}^{n}}\left\{ \omega:B(\widehat{g}(\omega);P_{0})=k\right\} \rightarrow1$
as $n\rightarrow\infty$.

Following the notation in Assumption \ref{cond:conv constraint},
denote the set of policies where the constraints bind exactly under
the limit distribution $P_{0}$ by
\begin{equation*}
\mathcal{G}_{0}=\{g\in\mathcal{G}:B(g;P_{0})=k\}.\label{eq:exactly binding DGP}
\end{equation*}

Under Assumption \ref{cond:conv constraint}, the sequence $P_{h_{n}}^{n}$
is contiguous with respect to the sequence $P_{0}^{n}$ , which means
$P_{0}^{n}(A_{n})\rightarrow0$ implies $P_{h_{n}}^{n}(A_{n})\rightarrow0$
for every sequence of measurable sets $A_{n}$ on $\mathcal{A}^{n}$.
Then $\Pr_{P_{0}^{n}}\left\{ A_{n}:B(\widehat{g}(A_{n});P_{0})=k\right\} \rightarrow1$
implies there exists an $N(u)$ such that for all $n\geq N(u)$, we
have $\Pr_{P_{h_{n}}^{n}}\left\{ A_{n}:B(\widehat{g}(A_{n});P_{0})=k\right\} \geq1-u.$
That is, with high probability, the statistical rule $\widehat{g}$
selects policies from $\mathcal{G}_{0}$ based on the observed sample
distributed according to $P_{h_{n}}^{n}$. Recall Assumption \ref{cond:conv constraint}
implies for all $g\in\mathcal{G}_{0}$, for any sample size $n$,
we have $B(g;P_{h_{n}})-k>c/\sqrt{n}.$ Thus this statistical rule
cannot uniformly satisfy the constraint since with sample size $n\geq N(u)$,
we have

\begin{equation*}
\sup_{P\in\mathcal{P}}\Pr_{P^{n}}\left\{ B(\widehat{g};P)>k\right\} \geq \Pr_{P_{h_{n}}^{n}}\left\{ B(\widehat{g};P_{h_{n}})>k\right\} \geq1-u
\end{equation*}
\end{proof}
\begin{proof}
\textbf{Proof of Proposition \ref{prop: EWM no ptwise mistake}.}

By assumption, the probability an individual takes up the
treatment for $X\in[\underline{t},\overline{t}]$ is zero but between
zero and one otherwise. Then the budget function $B(t;P)\coloneqq \E_{P}[C\cdot\mathbf{1}\{X\leq t\}]$ is flat in the interval $[\underline{t},\overline{t}]$ but strictly increasing otherwise. Since $\tau>0$ almost surely, the welfare function $W(t;P)\coloneqq \E_{P}[\tau\cdot\mathbf{1}\{X\leq t\}]$ is strictly increasing in $t$ and the constrained optimal policy satisfies Assumption~\ref{assu:exactly binding optimal} and is binding. Let $\inf_{B(g;P)<k}W(g_{P}^{\ast};P)-W(g;P)=\E_{P}[\tau\cdot\mathbf{1}\{\underline{t}\leq X\leq \overline{t}\}] =\epsilon$ be the smallest amount of welfare loss from missing the binding solution $g_{P}^{\ast}$. Then $\epsilon>0$ and satisfies Assumption~\ref{assu:gain when exactly binding}.

The population problem is
\[
\max_{t}\ W(t;P)\text{ subject to }B(t;P)\leq k,
\]
where $B(t;P)=k$ for $t\in[\underline{t},\overline{t}]$ by assumption.
The constrained optimal threshold is therefore the highest threshold
where the constraint is satisfied exactly i.e. $t^{\ast}=\overline{t}.$
This also implies $C\cdot\mathbf{1}\{X\leq t\}\sim Bernoulli(k)$
for $t\in[\underline{t},\overline{t}]$. 

The sample-analog rule solves the following sample problem 
\[
\max_{t}\frac{1}{n}\sum_{i}\tau_{i}\cdot\mathbf{1}\{X_{i}\leq t\}\text{ subject to }\widehat{B}_{n}(t)\leq k
\]
 for $\widehat{B}_{n}(t)\coloneqq\frac{1}{n}\sum_{i}C_{i}\cdot\mathbf{1}\{X_{i}\leq t\}$.
Given that $\tau>0$ almost surely, the sample-analog rule equivalently
solves $\max_{t}\widehat{B}_{n}(t)\text{ subject to }\widehat{B}_{n}(t)\leq k$.
However, the solution is not unique because $\widehat{B}_{n}(t)$
is a step function. To be conservative, let the sample-analog rule
be the smallest possible threshold to maximize $\widehat{B}_{n}(t)$:
\[
\widehat{t}=\min\left\{ \arg\max_{t}\{\widehat{B}_{n}(t)\text{ subject to }\widehat{B}_{n}(t)\leq k\}\right\} .
\]
Note that we can also write $\widehat{B}_{n}(t)=\frac{1}{n}\sum_{C_{i}=1}\mathbf{1}\{X_{i}\leq t\}$,
which makes it clear that $\widehat{t}$ corresponds to ranking $X_{i}$
among individuals with $C_{i}=1$, and then picking the lowest threshold
such that we assign treatment to the first $\left\lfloor k\cdot n\right\rfloor $
individuals. This also means if in the sample few individuals take
up the treatment such that $\frac{1}{n}\sum_{i}C_{i}\leq k$, we can
have a sample-analog rule that treats everyone up to $\max_{C_{i}=1}X_{i}$.
Taken together both scenarios, we note the sample-analog rule implies
the treated share in the sample is equal to
\[
\widehat{B}_{n}(\widehat{t})\coloneqq\frac{1}{n}\sum_{i}C_{i}\cdot\mathbf{1}\{X_{i}\leq\widehat{t}\}=\min\left\{ \frac{1}{n}\sum_{i}C_{i},\frac{\left\lfloor k\cdot n\right\rfloor }{n}\right\} .
\]
Note that 
\[
\begin{cases}
\widehat{t}>\overline{t}\Leftrightarrow & B(\widehat{t};P)>B(\overline{t};P)\\
\widehat{t}<\underline{t}\Leftrightarrow & B(\widehat{t};P)<B(\underline{t};P)\Leftrightarrow W(\widehat{t};P)<W(\underline{t};P)
\end{cases}
\]
which means whenever $\widehat{t}>\overline{t}$, the sample-analog
rule violates the constraint in the population as $B(\overline{t};P)=k$;
whenever $\widehat{t}<\underline{t}$, the sample-analog rule achieves
strictly less welfare than $t^{\ast}$ in the population because $W(\underline{t};P)$
is strictly less than $W(t^{\ast};P)$. We next derive the limit probability
for these two events. Applying Law of Total Probability, we have
\begin{align}
 & \Pr_{P^{n}}\left\{ \widehat{B}_{n}(\overline{t})<\widehat{B}_{n}(\widehat{t})\right\} \nonumber \\
= & \Pr_{P^{n}}\left\{ \widehat{B}_{n}(\overline{t})<\frac{\left\lfloor k\cdot n\right\rfloor }{n}\text{ and }\frac{\left\lfloor k\cdot n\right\rfloor }{n}<\frac{1}{n}\sum_{i}C_{i}\right\} \nonumber \\
 & +\Pr_{P^{n}}\left\{ \widehat{B}_{n}(\overline{t})<\frac{1}{n}\sum_{i}C_{i}\text{ and }\frac{\left\lfloor k\cdot n\right\rfloor }{n}\geq\frac{1}{n}\sum_{i}C_{i}\right\} \nonumber \\
\geq & \Pr_{P^{n}}\left\{ \widehat{B}_{n}(\overline{t})<\frac{\left\lfloor k\cdot n\right\rfloor }{n}\text{ and }\frac{\left\lfloor k\cdot n\right\rfloor }{n}<\frac{1}{n}\sum_{i}C_{i}\right\} \nonumber \\
\geq & \Pr_{P^{n}}\left\{ \widehat{B}_{n}(\overline{t})<\frac{\left\lfloor k\cdot n\right\rfloor }{n}\right\} +\Pr_{P^{n}}\left\{ \frac{\left\lfloor k\cdot n\right\rfloor }{n}<\frac{1}{n}\sum_{i}C_{i}\right\} -1\label{eq:sample-analog-mistake-expand}
\end{align}
For the first term in (\ref{eq:sample-analog-mistake-expand}), we
have the following lower bound
\begin{align*}
 & \Pr_{P^{n}}\left\{ \frac{1}{n}\sum_{i}C_{i}\cdot\mathbf{1}\{X_{i}\leq\overline{t}\}\leq\frac{k\cdot n-1}{n}\right\} \\
= & \Pr_{P^{n}}\left\{ \sqrt{n}\left(\frac{1}{n}\sum_{i}C_{i}\cdot\mathbf{1}\{X_{i}\leq\overline{t}\}-k\right)\leq-\frac{1}{\sqrt{n}}\right\} \rightarrow0.5
\end{align*}
To see the convergence, we apply the Central Limit Theorem to the
LHS, and note that $-\frac{1}{\sqrt{n}}$ converges to zero. Denote
$p_{C}=\Pr\{C=1\}$. For the second term in (\ref{eq:sample-analog-mistake-expand}),
we have the following lower bound
\begin{align*}
 & \Pr_{P^{n}}\left\{ \frac{1}{n}\sum_{i}C_{i}\geq\frac{k\cdot n}{n}\right\} \\
= & \Pr_{P^{n}}\left\{ \sqrt{n}\left(\frac{1}{n}\sum_{i}C_{i}-p_{C}\right)\geq\sqrt{n}\cdot\left(k-p_{C}\right)\right\} \rightarrow1
\end{align*}
To see the convergence, we apply the Central Limit Theorem to the
LHS, and note that $\sqrt{n}\cdot\left(k-p_{C}\right)$ diverges to
$-\infty$ for $p_{C}>k$. We thus conclude
\[
\lim_{n\rightarrow\infty}\Pr_{P^{n}}\left\{ B(\widehat{t};P)>B(\overline{t};P)\right\} \geq0.5
\]
which proves $\widehat{t}$ is not pointwise asymptotically feasible
under the distribution $P$.

Similar argument shows $\Pr_{P^{n}}\left\{ \widehat{B}_{n}(\underline{t})>\widehat{B}_{n}(\widehat{t})\right\} $
has a limit of at least one half. We thus conclude
\[
\lim_{n\rightarrow\infty}\Pr_{P^{n}}\left\{ B(\widehat{t};P)<B(\underline{t};P)\right\} \geq0.5\Leftrightarrow\lim_{n\rightarrow\infty}\Pr_{P^{n}}\left\{ W(\widehat{t};P)<W(\underline{t};P)\right\} \geq0.5
\]
which proves that $\widehat{t}$ is not pointwise asymptotically welfare-efficient
under the distribution $P$. Since $B(\underline{t};P)=B(\overline{t};P)$,
we actually have
\[
\lim_{n\rightarrow\infty}\Pr_{P^{n}}\left\{ B(\widehat{t};P)>B(\overline{t};P)\right\} =\lim_{n\rightarrow\infty}\Pr_{P^{n}}\left\{ W(\widehat{t};P)<W(\underline{t};P)\right\} =0.5
\]
\end{proof}

\begin{proof}
\textbf{Proof of Corollary \ref{prop:analog violations}.} The event
in (\ref{eq:sample-analog-violation}) is equivalent to the event
that $\hat{\mathcal{G}}$ includes at least one criteria that violates
the budget constraint by $c$. The derivation for the bound therefore
is based on such an event:
\begin{align*}
 & \Pr_{P^{n}}\left\{ \min_{B(g;P)>k+c}\frac{\sqrt{n}\left(\widehat{B}_{n}(g)-k\right)}{\Sigma^{B}(g,g)^{1/2}}\leq0\right\} \\
= & \Pr_{P^{n}}\left\{ \min_{B(g;P)>k+c}\left\{ \frac{\sqrt{n}\left(\widehat{B}_{n}(g)-B(g;P)\right)}{\Sigma^{B}(g,g)^{1/2}}+\frac{\sqrt{n}\left(B(g;P)-k\right)}{\Sigma^{B}(g,g)^{1/2}}\right\} \leq0\right\} \\
\leq & \Pr_{P^{n}}\left\{ \min_{B(g;P)>k+c}\frac{\sqrt{n}\left(\widehat{B}_{n}(g)-B(g;P)\right)}{\Sigma^{B}(g,g)^{1/2}}\leq-\min_{B(g;P)>k+c}\frac{\sqrt{n}\left(B(g;P)-k\right)}{\Sigma^{B}(g,g)^{1/2}}\right\} \\
\leq & \Pr_{P^{n}}\left\{ \min_{B(g;P)>k+c}\frac{\sqrt{n}\left(\widehat{B}_{n}(g)-B(g;P)\right)}{\Sigma^{B}(g,g)^{1/2}}\leq-\frac{\min_{B(g;P)>k+c}\sqrt{n}\left(B(g;P)-k\right)}{\max_{B(g;P)>k+c}\Sigma^{B}(g,g)^{1/2}}\right\} \\
< & \Pr_{P^{n}}\left\{ \min_{B(g;P)>k+c}\frac{\sqrt{n}\left(\widehat{B}_{n}(g)-B(g;P)\right)}{\Sigma^{B}(g,g)^{1/2}}\leq\frac{-\sqrt{n}\cdot c}{\max_{B(g;P)>k+c}\Sigma^{B}(g,g)^{1/2}}\right\} \\
< & \Pr_{P^{n}}\left\{ \min_{g\in\mathcal{G}}\frac{\sqrt{n}\left(\widehat{B}_{n}(g)-B(g;P)\right)}{\Sigma^{B}(g,g)^{1/2}}\leq\frac{-\sqrt{n}\cdot c}{\max_{g\in\mathcal{G}}\Sigma^{B}(g,g)^{1/2}}\right\} 
\end{align*}
Under Assumption \ref{cond:known ps Donsker}, the empirical process
$\left\{ \sqrt{n}\left(\widehat{B}_{n}(g)-B(g;P)\right)\right\} $
converges to a Gaussian process $G_{P}^{B}$ for $G_{P}^{B}(\cdot)\sim\mathcal{GP}(0,\Sigma_{P}^{B}(\cdot,\cdot))$
and we have a consistent covariance estimate $\widehat{\Sigma}^{B}(\cdot,\cdot)$,
which proves the corollary. %

\end{proof}

\begin{proof}
\textbf{Proof of Theorem \ref{thm: uniform size}.} By construction,
the limit probability for any policy in $\hat{\mathcal{G}}_{\alpha}$
to violate the budget constraint is
\begin{align*}
 & \Pr_{P^{n}}\{\exists g:g\in\hat{\mathcal{G}}_{\alpha}\text{ and }B(g;P)>k\}=\Pr_{P^{n}}\{\min_{B(g;P)>k}\frac{\sqrt{n}\left(\widehat{B}_{n}(g)-k\right)}{\widehat{\Sigma}^{B}(g,g)^{1/2}}\leq c_{\alpha}\}\\
\leq & \Pr_{P^{n}}\{\min_{B(g;P)>k}\frac{\sqrt{n}\left(\widehat{B}_{n}(g)-B(g;P)\right)}{\widehat{\Sigma}^{B}(g,g)^{1/2}}\leq c_{\alpha}\}\\
\leq & \Pr_{P^{n}}\{\min_{g\in\mathcal{G}}\frac{\sqrt{n}\left(\widehat{B}_{n}(g)-B(g;P)\right)}{\widehat{\Sigma}^{B}(g,g)^{1/2}}\leq c_{\alpha}\}
\end{align*}
Under Assumption \ref{cond:known ps Donsker}, uniformly over $P\in\mathcal{P}$,
the empirical process $\left\{ \sqrt{n}\left(\widehat{B}_{n}(g)-B(g;P)\right)\right\} $
converges to a Gaussian process $G_{P}^{B}$ for $G_{P}^{B}(\cdot)\sim\mathcal{GP}(0,\Sigma_{P}^{B}(\cdot,\cdot))$
and we have a consistent covariance estimate $\widehat{\Sigma}^{B}(\cdot,\cdot)$.
Then by the definition of $c_{\alpha}$, we have
\begin{align*}
 & \limsup_{n\rightarrow\infty}\sup_{P\in\mathcal{P}}\Pr_{P^{n}}\{\min_{g\in\mathcal{G}}\frac{\sqrt{n}\left(\widehat{B}_{n}(g)-B(g;P)\right)}{\widehat{\Sigma}^{B}(g,g)^{1/2}}\leq c_{\alpha}\}\\
= & \sup_{P\in\mathcal{P}}\Pr_{P^{n}}\{\inf_{\text{ }g\in\mathcal{G}}\frac{G_{P}^{B}(g)}{\Sigma_{P}^{B}(g,g)^{1/2}}\leq c_{\alpha}\}=\alpha.
\end{align*}
\end{proof}
\begin{proof}
\textbf{Proof of  Corollary~\ref{cor:optimal mistake control}
}Uniformly asymptotic feasibility
follows from Theorem \ref{thm: uniform size}, replacing $\alpha$
with $\alpha_{n}$ in its proof. Specifically, by construction we
have 
\[
\limsup_{n\rightarrow\infty}\sup_{P\in\mathcal{P}}\Pr_{P^{n}}\{B(\hat{g}_{\alpha_n};P)>k\}\leq\limsup_{n\rightarrow\infty}\sup_{P\in\mathcal{P}}\Pr_{P^{n}}\{\exists g:g\in\hat{\mathcal{G}}_{\alpha_{n}}\text{ and }B(g;P)>k\}\leq\alpha_{n}.
\]
For the second part, we decompose
the welfare loss into
\begin{align*}
 & W(g_{P}^{\ast};P)-W(\hat{g}_{\alpha_n};P)\\
= & W(g_{P}^{\ast};P)-\widehat{W}_{n}(\hat{g}_{\alpha_n})+\widehat{W}_{n}(\hat{g}_{\alpha_n})-W(\hat{g}_{\alpha_n};P)\\
\leq & \sup_{g\in\mathcal{G}}2\vert W(g;P)-\widehat{W}_{n}(g)\vert+\widehat{W}_{n}(g_{P}^{\ast})-\widehat{W}_{n}(\hat{g}_{\alpha_n})
\end{align*}
Under the event that $g_{P}^{\ast}\in\hat{\mathcal{G}}_{\alpha_{n}}$
, we are guaranteed that the last term is non-positive. This happens
with probability
\[
\Pr_{P^{n}}\{\frac{\sqrt{n}\left(\widehat{B}_{n}(g_{P}^{\ast})-k\right)}{\widehat{\Sigma}^{B}(g_{P}^{\ast},g_{P}^{\ast})^{1/2}}\leq c_{\alpha_{n}}\}\stackrel{a}{\sim}\Phi\left(c_{\alpha_{n}}+\frac{\sqrt{n}\left(k-B(g_{P}^{\ast};P)\right)}{\Sigma_{P}^{B}(g_{P}^{\ast},g_{P}^{\ast})^{1/2}}\right)\rightarrow1
\]
for $B(g_{P}^{\ast};P)$ strictly below the threshold. Since $c_{\alpha_{n}}$
diverges to $-\infty$ at a rate slower $\sqrt{n}$, the term in the
parenthesis diverges to $\infty$ as $n\rightarrow\infty$. We then
apply Lemma \ref{lem:known ps Glivenko=002013Cantelli} and \ref{lem:ps vanish estimation error}
in the Appendix (proved in Section \ref{subsec:Results-with-estimated}
under primitive conditions that lead to Assumption \ref{cond:known ps Donsker}),
which shows $\sup_{g\in\mathcal{G}}\vert W(g;P)-\widehat{W}_{n}(g)\vert$
converges to zero in probability uniformly. We then conclude $\hat{g}_{\alpha_n}$
is asymptotically welfare-efficient uniformly under  $\mathcal{P}$.
\end{proof}

\begin{proof}
\textbf{Proof of  Corollary~\ref{cor:positive cost}
}Uniformly asymptotic feasibility
follows exactly as in the Proof of Corollary~\ref{cor:optimal mistake control}.  To prove uniformity in asymptotic welfare-efficiency, I slightly abuse notation by indexing the policy class by $\theta$
in the welfare and budget functions below.
Since the policy class includes thresholding on continuously distributed
$X$, and by assumption the conditional benefit and cost are positive, the welfare function $W(\theta;P)$ and the budget function $B(\theta;P)$
are both continuous and increases in $\theta$. That is,
whenever $\theta\prec\theta'$ where $\prec$ denotes element-wise
inequality, we have $W(\theta;P)<W(\theta';P)$ and $B(\theta;P)<B(\theta';P).$ Therefore, at the constrained optimal threshold $\theta_{P}^{\ast}$, the constraint is binding $B(\theta_{P}^{\ast};P)=k.$  By the
continuity and strict monotonicity of $W(\theta;P)$, for any $\epsilon>0$,
there exists $\theta_{P,\epsilon}\prec\theta_{P}^{\ast}$ such that
$W(\theta_{P,\epsilon};P)=W(\theta_{P}^{\ast};P)-\epsilon$. 

Construct a sequence $b_{n}=-2c_{\alpha_{n}}/\sqrt{n}$,
which decreases to zero since $c_{\alpha_{n}}$ diverges to $-\infty$
at a rate slower $\sqrt{n}$. Consider $B(\theta_{P}^{\ast};P)-B(\theta_{P,\epsilon};P)=k-B(\theta_{P,\epsilon};P)$,
which is strictly positive for all $P\in\mathcal{P}$ because by assumption the
conditional cost is bounded away from zero. Therefore there exists
$N$ large enough such that for all $P\in\mathcal{P}$ we have
\begin{equation}
k-B(\theta_{P,\epsilon};P)\geq b_{n}\cdot\Sigma_{P}^{B}(\theta_{P,\epsilon},\theta_{P,\epsilon})^{1/2}>0,\ \forall n\geq N.\label{eq:cost bdd below}
\end{equation}

We decompose the welfare loss into
\begin{align*}
 & W(\theta_{P}^{\ast};P)-W(\hat{\theta}_{\alpha_{n}};P)\\
= & W(\theta_{P}^{\ast};P)-W(\theta_{P,\epsilon};P)+W(\theta_{P,\epsilon};P)-W(\hat{\theta}_{\alpha_{n}};P)\\
= & \epsilon+W(\theta_{P,\epsilon};P)-\widehat{W}_{n}(\theta_{P,\epsilon})+\widehat{W}_{n}(\theta_{P,\epsilon})-\widehat{W}_{n}(\hat{\theta}_{\alpha_{n}})+\widehat{W}_{n}(\hat{\theta}_{\alpha_{n}})-W(\hat{\theta}_{\alpha_{n}};P)\\
\leq & \epsilon+\sup_{g\in\mathcal{G}}2\vert W(g;P)-\widehat{W}_{n}(g)\vert+\widehat{W}_{n}(\theta_{P,\epsilon})-\widehat{W}_{n}(\hat{\theta}_{\alpha_{n}})
\end{align*}
Under the event that $\theta_{P,\epsilon}\in\hat{\mathcal{G}}_{\alpha_{n}}$, we are guaranteed that the last term is non-positive. This happens
with probability
\[
\Pr_{P^{n}}\{\frac{\sqrt{n}\left(\widehat{B}_{n}(\theta_{P,\epsilon})-k\right)}{\widehat{\Sigma}^{B}(\theta_{P,\epsilon},\theta_{P,\epsilon})^{1/2}}\leq c_{\alpha_{n}}\}\stackrel{a}{\sim}\Phi\left(c_{\alpha_{n}}+\frac{\sqrt{n}\left(k-B(\theta_{P,\epsilon};P)\right)}{\Sigma_{P}^{B}(\theta_{P,\epsilon},\theta_{P,\epsilon})^{1/2}}\right).
\]
From (\ref{eq:cost bdd below}), we have
\[
\lim_{n\rightarrow\infty}\inf_{P\in\mathcal{P}}\Phi\left(c_{\alpha_{n}}+\frac{\sqrt{n}\left(k-B(\theta_{P,\epsilon};P)\right)}{\Sigma_{P}^{B}(\theta_{P,\epsilon},\theta_{P,\epsilon})^{1/2}}\right)\geq\lim_{n\rightarrow\infty}\Phi\left(c_{\alpha_{n}}-2c_{\alpha_{n}}\right)=\lim_{n\rightarrow\infty}\Phi\left(-c_{\alpha_{n}}\right)=1
\]
as $c_{\alpha_{n}}$ diverges to $-\infty$ as $n\rightarrow\infty$.
We then apply Lemma \ref{lem:known ps Glivenko=002013Cantelli} and
\ref{lem:ps vanish estimation error} in the Appendix (proved
in Section \ref{subsec:Results-with-estimated} under primitive conditions
that lead to Assumption \ref{cond:known ps Donsker}), which shows
$\sup_{g\in\mathcal{G}}\vert W(g;P)-\widehat{W}_{n}(g)\vert$ converges
to zero in probability uniformly over $P\in\mathcal{P}$. We then
conclude $\hat{g}_{\alpha_{n}}$ is uniformly asymptotically welfare-efficient
under  $\mathcal{P}$.
\end{proof}
\begin{proof}
\textbf{Proof of Lemma \ref{lem:hinge higher value}.} By definition
\begin{align*}
W(g_{P}^{\ast};P)= & \max_{B(g;P)\leq k}W(g;P)=\max_{B(g;P)\leq k}W(g;P)-\frac{\overline{\lambda}}{r}\cdot(B(g;P)-k)_{+}\\
 & \leq\max_{g\in\mathcal{G}}W(g;P)-\frac{\overline{\lambda}}{r}\cdot(B(g;P)-k)_{+}\\
 & =W(\widetilde{g}_{P};P)-\frac{\overline{\lambda}}{r}\cdot(B(\widetilde{g}_{P};P)-k)_{+}\leq W(\widetilde{g}_{P};P)
\end{align*}
and suppose $B(\widetilde{g}_{P};P)>k$, we have the following upper
bound for violations to the budget constraint
\[
B(\widetilde{g}_{P};P)-k\leq\frac{r\cdot (W(\widetilde{g}_{P};P)-W(g_{P}^{\ast};P))}{\overline{\lambda}}.
\]
\end{proof}
\begin{proof}
    \textbf{Proof of Lemma~\ref{lem:hinge consistency}}
    
Recall the new objective function with $V(g;P)=W(g;P)-\frac{\overline{\lambda}}{r}\cdot(B(g;P)-k)_{+}$,
whose maximizer is $\widetilde{g}_{P}$. Recall $\widehat{V}_{n}(g)=\widehat{W}_{n}(g)-\frac{\overline{\lambda}}{r}\cdot(\widehat{B}_{n}(g)-k)_{+}$
is the sample-analog, whose maximizer is $\widehat{g}_{\text{tradeoff}}$.
Apply uniform deviation bound to the difference between the value
under $\widehat{g}_{\text{tradeoff}}$ and $\widetilde{g}_{P}$, we
have
\begin{align*}
 & V(\widetilde{g}_{P};P)-V(\widehat{g}_{\text{tradeoff}};P)\\
= & V(\widetilde{g}_{P};P)-\widehat{V}_{n}(\widehat{g}_{\text{tradeoff}})+\widehat{V}_{n}(\widehat{g}_{\text{tradeoff}})-V(\widehat{g}_{\text{tradeoff}};P)\\
\leq & V(\widetilde{g}_{P};P)-\widehat{V}_{n}(\widetilde{g}_{P})+\widehat{V}_{n}(\widehat{g}_{\text{tradeoff}})-V(\widehat{g}_{\text{tradeoff}};P)\\
\leq & 2\sup_{g}\left|V(g;P)-\widehat{V}_{n}(g)\right|\\
= & 2\sup_{g}\left|W(g;P)-\widehat{W}_{n}(g)-\frac{\overline{\lambda}}{r}\cdot(\max\{B(g;P)-k,0\}-\max\{\widehat{B}_{n}(g)-k,0\})\right|\\
\leq & 2\sup_{g}\left|\widehat{W}_{n}(g)-W(g;P)\right|+2\frac{\overline{\lambda}}{r}\cdot\sup_{g}\left|\max\{\widehat{B}_{n}(g)-k,0\}-\max\{B(g;P)-k,0\}\right|\\
\leq & 2\sup_{g}\left|\widehat{W}_{n}(g)-W(g;P)\right|+2\frac{\overline{\lambda}}{r}\cdot\sup_{g}\left|\widehat{B}_{n}(g)-B(g;P)\right|
\end{align*}
The last line uses the fact that $\left|\max\{a,0\}-\max\{b,0\}\right|\leq\left|a-b\right|$.
Both terms, $\sup_{g}\left|\widehat{W}_{n}(g)-W(g;P)\right|$ and
$\sup_{g}\left|\widehat{B}_{n}(g)-B(g;P)\right|$ converge to zero
in probability under Assumption \ref{cond:known ps Donsker}. Specifically,
Lemma \ref{lem:known ps Glivenko=002013Cantelli} and \ref{lem:ps vanish estimation error}
in the Appendix (proved in Section \ref{subsec:Results-with-estimated}
under primitive conditions that lead to Assumption \ref{cond:known ps Donsker})
imply the uniform convergence in probability
\[
\sup_{P\in\mathcal{P}}\sup_{g\in\mathcal{G}}\left|\widehat{W}_{n}(g)-W(g;P)\right|\rightarrow_{p}0,\ \sup_{P\in\mathcal{P}}\sup_{g\in\mathcal{G}}\left|\widehat{B}_{n}(g)-B(g;P)\right|\rightarrow_{p}0.
\]
At the same time, by definition we have $V(\widetilde{g}_{P};P)-V(\widehat{g}_{\text{tradeoff}};P)\geq0$.
We thus conclude 
\[
\sup_{P\in\mathcal{P}}\left|V(\widehat{g}_{\text{tradeoff}};P)-V(\widetilde{g}_{P};P)\right|\rightarrow_{p}0.
\]
\end{proof}
\begin{proof}
\textbf{Proof of Theorem \ref{thm:hinge tradeoff}. }

By Lemma \ref{lem:hinge higher value}, we have $W(g_{P}^{\ast};P)\leq V(\widetilde{g}_{P};P)\leq W(\widetilde{g}_{P};P)$
  for any $P\in\mathcal{P}$.
Putting these together with Lemma~\ref{lem:hinge consistency}, we have
\begin{align*}
 & \inf_{P\in\mathcal{P}}\left\{ W(\widehat{g}_{\text{tradeoff}};P)-W(g_{P}^{\ast};P)\right\} \\
\geq & \inf_{P\in\mathcal{P}}\left\{ V(\widehat{g}_{\text{tradeoff}};P)-W(g_{P}^{\ast};P)\right\} \\
\geq & \inf_{P\in\mathcal{P}}\left\{ V(\widehat{g}_{\text{tradeoff}};P)-V(\widetilde{g}_{P};P)\right\} +\inf_{P\in\mathcal{P}}\left\{ V(\widetilde{g}_{P};P)-W(g_{P}^{\ast};P)\right\} \\
\geq & \inf_{P\in\mathcal{P}}\left\{ V(\widehat{g}_{\text{tradeoff}};P)-V(\widetilde{g}_{P};P)\right\} \rightarrow_{p}0.
\end{align*}
which proves uniform asymptotic welfare-efficiency.

By the uniform deviation bound used in the proof for Lemma~\ref{lem:hinge consistency}, applied to the difference between the value
under $\widehat{g}_{\text{tradeoff}}$ and $g_{P}^{\ast}$, we have that uniformly
with probability approaching one
\begin{align*}
 & V(g_{P}^{\ast};P)-V(\widehat{g}_{\text{tradeoff}};P)\leq0
\end{align*}
which is equivalent to $(B(\widehat{g}_{\text{tradeoff}};P)-k)_{+}\leq\frac{r\cdot(W(\widehat{g}_{\text{tradeoff}};P)-W(g_{P}^{\ast};P))}{\overline{\lambda}}$.
\end{proof}

\section{Primitive assumptions and auxiliary lemmas\label{sec:Primitive-assumptions-and}}

I first prove the optimal rule that solves the population constrained
optimization problem takes the form of threshold. In Section \ref{subsec:Primitive-for-contiguity},
I first provide primitive assumptions on the class of DGPs. I then
prove in Lemma \ref{thm:DQM contiguity}, which establishes that these
primitive assumptions imply Assumption \ref{cond:conv constraint}.

In Section \ref{subsec:Results-with-estimated}, I verify Assumption
\ref{cond:known ps Donsker} for settings where the observed sample
comes from an RCT or an observational study, and the propensity score
can be estimated efficiently based on parametric regressions.

\subsection{Constrained optimal rule without functional form restriction}

The population problem is to find rules based on $X_{i}$ that solves
\[
\max_{g:\mathcal{X}\rightarrow\{0,1\}}\E[\tau_{i}g(X_{i})]\text{ subject to }\E[C_{i}g(X_{i})]<k
\]

By Law of Iterated Expectation, we can write the constrained optimization
problem as 
\begin{align*}
\max_{g:\mathcal{X}\rightarrow\{0,1\}}\E[\gamma(X_{i})g(X_{i})]\text{ subject to }\E[r(X_{i})g(X_{i})] & <k
\end{align*}
where $\gamma(X_{i})=\E[\tau_{i}\mid X_{i}]$ and $r(X_{i})=\E[C_{i}\mid X_{i}]$. 
\begin{claim}
Let $d\mu=r(x)f(x)dx$ denote the positive measure. The constrained
optimization problem is equivalent to 
\begin{align*}
 & \max_{g:\mathcal{X}\rightarrow\{0,1\}}\int\frac{\gamma(x)}{r(x)}g(x)d\mu\text{ subject to }\int g(x)d\mu=k
\end{align*}
Let $X^{\ast}$ be the support of the solution $g^{\ast}$. It will
take the form of $X^{\ast}=\{x:\frac{\gamma(x)}{r(x)}>c\}$ where
$c$ is chosen so that $\mu(X^{\ast})=k$. 
\end{claim}
\begin{proof}
Let $X$ be the support of any $g\neq g^{\ast}$ with $\mu(X)=k$.
Then the objective function associated $g$ is
\begin{align*}
\int_{X^{\ast}}\frac{\gamma}{r}d\mu-\int_{X}\frac{\gamma}{r}d\mu & =\int\frac{\gamma(x)}{r(x)}\mathbf{1}\{x\in X^{\ast}\}d\mu-\int\frac{\gamma(x)}{r(x)}\mathbf{1}\{x\in X\}d\mu\\
 & =\int\frac{\gamma(x)}{r(x)}\left(\mathbf{1}\{x\in X^{\ast}\backslash X\}-\mathbf{1}\{x\in X\backslash X^{\ast}\}\right)d\mu
\end{align*}
By definition of $X^{\ast}$, we have $\frac{\gamma(x)}{r(x)}>c$
for $x\in X^{\ast}\backslash X$ and $\frac{\gamma(x)}{r(x)}<c$ for
$x\in X\backslash X^{\ast}$. Also note that $\mu$ is a positive
measure. Then the above difference is lower bounded by
\[
\int\frac{\gamma(x)}{r(x)}\left(\mathbf{1}\{x\in X^{\ast}\backslash X\}-\mathbf{1}\{x\in X\backslash X^{\ast}\}\right)d\mu\geq c\int\left(\mathbf{1}\{x\in X^{\ast}\backslash X\}-\mathbf{1}\{x\in X\backslash X^{\ast}\}\right)d\mu\geq0
\]
as by construction, we have \[\int\left(\mathbf{1}\{x\in X^{\ast}\backslash X\}-\mathbf{1}\{x\in X\backslash X^{\ast}\}\right)d\mu=\int\left(\mathbf{1}\{x\in X^{\ast}\}-\mathbf{1}\{x\in X\}\right)d\mu=0\]
since $\mu(X^{\ast})=\mu(X)=k$.
\end{proof}

\subsection{Primitive assumptions for contiguity\label{subsec:Primitive-for-contiguity}}
\begin{assumption}
\label{assu:DQM}Assume the class of DGPs $\{P_{\theta}:\theta\in\Theta\}$
has densities $p_{\theta}$ with respect to some measure $\mu$. Assume
$P_{\theta}$ is DQM at $P_{0}$ i.e. $\exists\dot{\ell}_{0}$ s.t.
$\int[\sqrt{p_{h}}-\sqrt{p_{0}}-\frac{1}{2}h'\dot{\ell}_{0}\sqrt{p_{0}}]^{2}d\mu=o(\left\Vert h^{2}\right\Vert )$
for $h\rightarrow0$. 
\end{assumption}
\begin{assumption}
\label{assu:regular}For all policies $g$, $B(g;P_{\theta})$ is
twice continuously differentiable in $\theta$ at $0$, and the derivatives
are bounded from above and away from zero within an open neighborhood
$\mathcal{N}_{\theta}$ of zero uniformly over $g\in\mathcal{G}$. 
\end{assumption}
\begin{lem}
\label{thm:DQM contiguity}Under Assumption \ref{assu:DQM}, the class
$\mathcal{P}$ includes a sequence of data distribution $\{P_{h_{n}}\}$
that is contiguous to $P_{0}$ for every $h_{n}$ satisfying $\sqrt{n}h_{n}\rightarrow h$
e.g. take $h_{n}=h/\sqrt{n}$. This proves the first part of Assumption
\ref{cond:conv constraint}. Suppose further Assumption \ref{assu:regular}
holds, then there exists some $h$ for the second part of Assumption
\ref{cond:conv constraint} to hold,.
\end{lem}
\begin{proof}
\textbf{Proof of Lemma \ref{thm:DQM contiguity}.} By Theorem 7.2
of \citet{vaart_asymptotic_1998}, the log likelihood ratio process
converges under $P_{0}$ (denoted with $\overset{p_{0}}{\rightsquigarrow}$)
to a normal experiment 
\begin{align}
\log\prod_{i=1}^{n}\frac{p_{h_{n}}}{p_{0}}(A_{i}) & =\frac{1}{\sqrt{n}}\sum_{i=1}^{n}h'\dot{\ell}_{0}(A_{i})-\frac{1}{2}h'I_{0}h+o_{P_{0}}(1)\nonumber\\
 & \overset{p_{0}}{\rightsquigarrow}\mathcal{N}\left(-\frac{1}{2}h'I_{0}h,h'I_{0}h\right)\label{eq:Gaussian exp}
\end{align}
where $\dot{\ell}_{0}$ is the score and $I_{P_{0}}=E_{P_{0}}[\dot{\ell}_{0}(A_{i})\dot{\ell}_{0}(A_{i})']$
exists. The convergence in distribution of the log likelihood ratio
to a normal with mean equal to $-\frac{1}{2}$ of its variance in
(\ref{eq:Gaussian exp}) implies mutual contiguity $P_{0}^{n}\vartriangleleft\vartriangleright P_{h_{n}}^{n}$
by Le Cam's first lemma (see Example 6.5 of \citet{vaart_asymptotic_1998}).
This proves the first part of the lemma.

By Taylor's theorem with remainder we have for each policy $g$
\[
B(g;P_{h/\sqrt{n}})-B(g;P_{0})=\frac{h'}{\sqrt{n}}\frac{\partial B(g;P_{0})}{\partial\theta}+\frac{1}{2}\frac{1}{n}h'\frac{\partial^{2}B(g;P_{\tilde{\theta}_{n}})}{\partial\theta\partial\theta'}h
\]
 where $\tilde{\theta}_{n}$ is a sequence of values with $\tilde{\theta}_{n}\in[0,h/\sqrt{n}]$
that can depend on $g$. Take $h$ so that the first term is positive
for policies with $B(g;P_{0})=k$. Such $h$ exists because we assume
$\frac{\partial B(g;P_{0})}{\partial\theta}$ is bounded away from
zero. For $g\in\mathcal{G}_{0}$ where $\mathcal{G}_{0}=\left\{ g:B(g;P_{0})=k\right\} $,
the constraints are violated under $P_{h_{n}}$ and furthermore (multiplying
by $\sqrt{n}$) 
\begin{align*}
\sqrt{n}\cdot(B(g;P_{h_{n}})-k) & >h'\frac{\partial B(g;P_{0})}{\partial\theta}>0
\end{align*}
for every $n$. This proves the second part of the lemma for $c=\inf_{g\in\mathcal{G}_{0}}\left|h'\frac{\partial B(g;P_{0})}{\partial\theta}\right|$
. 
\end{proof}

\subsection{Primitive assumptions and proofs for estimation quality \label{subsec:Results-with-estimated}}

In this section, I verify that $\widehat{W}_{n}(g)$ and $\widehat{B}_{n}(g)$
satisfy Assumption \ref{cond:known ps Donsker} under primitive assumptions
on the policy class and the OHIE. 

\begin{assumption}
\label{assu:VC class}VC-class: The policy class $\mathcal{G}$ has
a finite VC-dimension $v<\infty$.
\end{assumption}
To introduce the assumptions on the OHIE, I first recall the definitions
of components in $\widehat{W}_{n}(g)$ and $\widehat{B}_{n}(g)$:
\[
\widehat{W}_{n}(g)\coloneqq\frac{1}{n}\sum_{i}\tau_{i}^{\ast}\cdot g(X_{i}),\ \widehat{B}_{n}(g)\coloneqq\frac{1}{n}\sum_{i}C_{i}^{\ast}\cdot g(X_{i})
\]
 for the doubly-robust scores
\begin{align*}
\tau_{i}^{\ast} & =\widehat{\gamma}^{Y}(V_{i},1)-\widehat{\gamma}^{Y}(V_{i},0)+\widehat{\alpha}(V_{i},D_{i})\cdot\left(Y_{i}-\widehat{\gamma}^{Y}(V_{i},D_{i})\right)\\
C_{i}^{\ast} & =\widehat{\gamma}^{Z}(V_{i},1)+\frac{D_{i}}{\widehat{p}(V_{i})}\cdot\left(Z_{i}-\widehat{\gamma}^{Z}(V_{i},D_{i})\right)
\end{align*}
and observed characteristics $X_{i}$. Here $V_{i}$ collects the
confounders in OHIE, namely household size (number of adults entered
on the lottery sign-up form) and survey wave. Note that while $X_{i}$
can overlap with $V_{i}$, the policy $g(X_{i})$ needs not vary by
$V_{i}$. In the OHIE example, the policy is based on number of children
and income. However, conditional on household size and survey wave,
income and number of children is independent of the lottery outcome
in OHIE.

Recall that $\widehat{W}_{n}(g)$ and $\widehat{B}_{n}(g)$ are supposed
to approximate net benefit $\tau$ and net excess cost $C$ of Medicaid
eligibility. I provide more precise definitions for $\tau$ and
$C$ as the primitive assumptions are stated in terms of their components.

Let $Y(1)$ be the (potential) subjective health when one is given
Medicaid eligibility, and $Y(0)$ be the (potential) subjective health
when one is not given Medicaid eligibility. Recall the definition
for 
\[
\tau=Y(1)-Y(0)
\]
We only observe the actual subjective health $Y_{i}$. 

Let $M(1)$ be the (potential) enrollment in Medicaid when one is
given Medicaid eligibility, and $M(0)$ be the (potential) enrollment
in Medicaid when one is not given Medicaid eligibility. 

Let $H(1)$ be the (potential) health care utilization when one is
given Medicaid eligibility, and $H(0)$ be the (potential) health care utilization when one is not given Medicaid eligibility.  Even when given eligibility, one might
not enroll and thus incur zero cost to the government. So the per capita cost of Medicaid eligibility policy
$g(X)$ is 
\[
\E_{P}[C\cdot g(X)]\text{ where }C=M(1) \cdot H(1) .
\]

We only observe the actual health care utilization ($H_{i}=D_{i}H_{i}(1)+(1-D_{i})H_{i}(0)$)
and the actual Medicaid enrollment ($M_{i}$) in OHIE. We calculate
$Z_{i}=M_{i}\cdot H_{i}$. 
\begin{assumption}
\label{assu:known ps}Suppose for all $P\in\mathcal{P}$, the following
statements hold for the OHIE:

Independent characteristics: $\Pr\{D_{i}=1\mid V_{i},X_{i}\}=\Pr\{D_{i}=1\mid V_{i}\}$

Unconfoundedness: $(Y(1),Y(0),H(1),M(1))\perp D_{i}\vert V_{i}$.

Bounded attributes: the support of variables $X_{i}$, $Y_{i}$ and
$Z_{i}$ are bounded.

Strict overlap: There exist $\kappa\in(0,1/2)$ such that the propensity
score satisfies $p(v)\in[\kappa,1-\kappa]$ for all $v\in\mathcal{V}$.
\end{assumption}

\subsubsection{Uniform convergence of $\widehat{W}_{n}(\cdot)$ and $\widehat{B}_{n}(\cdot)$}

We want to show the recentered empirical processes $\widehat{W}_{n}(\cdot)$
and $\widehat{B}_{n}(\cdot)$ converge to mean-zero Gaussian processes
$G_{P}^{W}$ and $G_{P}^{B}$ with covariance functions $\Sigma_{P}^{W}(\cdot,\cdot)$
and $\Sigma_{P}^{B}(\cdot,\cdot)$ respectively uniformly over $P\in\mathcal{P}$.
The covariance functions are uniformly bounded, with diagonal entries
bounded away from zero uniformly over $g\in\mathcal{G}$. Take $\widehat{W}_{n}(\cdot)$
for example, the recentered empirical processes is

\[
\sqrt{n}\left(\frac{1}{n}\sum_{i}\tau_{i}^{\ast}\cdot g(X_{i})-\E_{P}[\tau\cdot g(X_{i})]\right)
\]
and can be expressed as the sum of two terms

\begin{equation*}
\frac{1}{\sqrt{n}}\sum_{i}(\tau_{i}^{\ast}-\widetilde{\tau}_{i})\cdot g(X_{i})+\sqrt{n}\left(\frac{1}{n}\sum_{i}\widetilde{\tau}_{i}\cdot g(X_{i})-\E_{P}[\tau\cdot g(X_{i})]\right).\label{eq:expand deviation}
\end{equation*}
Here $\widetilde{\tau}_{i}$ are the theoretical analogs
\begin{align*}
\widetilde{\tau}_{i} & =\gamma^{Y}(V_{i},1)-\gamma^{Y}(V_{i},0)+\alpha(V_{i},D_{i})\cdot\left(Y_{i}-\gamma^{Y}(V_{i},D_{i})\right)
\end{align*}
which is doubly-robust score with the theoretical propensity score
and the CEF. A similar expansion holds for $\widehat{B}_{n}(\cdot)$
involving the theoretical analog 
\[
\widetilde{C}_{i}=\gamma^{Z}(V_{i},1)+\frac{D_{i}}{p(V_{i})}\cdot\left(Z_{i}-\gamma^{Z}(V_{i},D_{i})\right).
\]

The following lemmas prove the uniform convergence of $\widehat{W}_{n}(\cdot)$
and $\widehat{B}_{n}(\cdot)$. 

The last part of the assumption is that we have a uniformly consistent
estimate for the covariance function. I argue the sample analog
\begin{align*}
\widehat{\Sigma}^{B}(g,g') & =\frac{1}{n}\sum_{i}(\widetilde{C}_{i})^{2}\cdot g(X_{i})\cdot g'(X_{i})-\left(\frac{1}{n}\sum_{i}\widetilde{C}_{i}\cdot g(X_{i})\right)\cdot\left(\frac{1}{n}\sum_{i}\widetilde{C}_{i}\cdot g'(X_{i})\right)
\end{align*}
is a pointwise consistent estimate for $\widehat{\Sigma}^{B}(g,g')$.
To see this, note that the covariance function $\Sigma_{P}^{B}(\cdot,\cdot)$
maps $\mathcal{G}\times\mathcal{G}$ to $\mathbb{R}$. Under Assumption
\ref{assu:known ps} that $\mathcal{G}$ is a VC-class, the product
$\mathcal{G}\times\mathcal{G}$ is also a VC-class (see e.g. Lemma
2.6.17 of \citet{van_der_vaart_weak_1996}). By a similar argument
leading to Lemma \ref{lem:known ps Glivenko=002013Cantelli} and \ref{lem:ps vanish estimation error}
, we conclude the uniform consistency of $\widehat{\Sigma}^{B}(\cdot,\cdot).$
\begin{lem}
\label{lem:vc-subgraph}Let $\mathcal{G}$ be a VC-class of subsets
of $\mathcal{X}$ with VC-dimension $v<\infty$. The following sets
of functions from $\mathcal{A}$ to $\mathbb{R}$
\begin{align*}
\mathcal{F}^{W} & =\{\widetilde{\tau}_{i}\cdot g(X_{i}):g\in\mathcal{G}\}\\
\mathcal{F}^{B} & =\{\widetilde{C}_{i}\cdot g(X_{i}):g\in\mathcal{G}\}
\end{align*}
are VC-subgraph class of functions with VC-dimension less than or
equal to $v$ for all $P\in\mathcal{P}$. For notational simplicity,
we suppress the dependence of $\mathcal{F}$ on $P$.
\end{lem}
\begin{lem}
\label{lem:identification}For all $P$ in the family $\mathcal{P}$
of distributions satisfying Assumption \ref{assu:known ps}, for all
$g\in\mathcal{G}$ we have
\begin{align*}
\E_{P}[\widetilde{\tau}_{i}\cdot g(X_{i})] & =W(g;P)\\
\E_{P}[\widetilde{C}_{i}\cdot g(X_{i})] & =B(g;P)
\end{align*}
\end{lem}
\begin{lem}
\label{lem:known ps Glivenko=002013Cantelli} Let $\mathcal{G}$ satisfy
Assumption \ref{assu:known ps} (VC). Let $U_{i}$ be a mean-zero
bounded random vector of fixed dimension i.e. there exists $M<\infty$
such that i.e. $U_{i}\in[-M/2,M/2]$ almost surely under all $P\in\mathcal{P}$.
Then the uniform deviation of the sample average of vanishes
\begin{align*}
\sup_{P\in\mathcal{P}}\sup_{g\in\mathcal{G}}\left|\frac{1}{n}\sum_{i}U_{i}\cdot g(X_{i})\right| & \rightarrow_{a.s.}\mathbf{0}.
\end{align*}
 
\end{lem}
\begin{lem}
\label{lem:known ps Donsker} Let $\mathcal{P}$ be a family of distributions
satisfying Assumption \ref{assu:known ps}. Let $\mathcal{G}$ satisfy
Assumption \ref{assu:known ps} (VC). Then $\mathcal{F}^{W}$ and
$\mathcal{F}^{B}$ are $P$-Donsker for all $P\in\mathcal{P}$. That
is, the empirical process indexed by $g\in\mathcal{G}$
\begin{align*}
\sqrt{n}\cdot(\frac{1}{n}\sum_{i}\widetilde{\tau}_{i}\cdot g(X_{i})-W(g;P))
\end{align*}
converge to a Gaussian process $\mathcal{GP}(0,\Sigma_{P}^{W}(\cdot,\cdot))$
uniformly in $P\in\mathcal{P}$, and the empirical process indexed
by $g\in\mathcal{G}$  
\[
\sqrt{n}\cdot(\frac{1}{n}\sum_{i}\widetilde{C}_{i}\cdot g(X_{i})-B(g;P))
\]
converge to a Gaussian process $\mathcal{GP}(0,\Sigma_{P}^{B}(\cdot,\cdot))$
uniformly in $P\in\mathcal{P}$.
\end{lem}
\begin{lem}
\label{lem:ps vanish estimation error} Let $\mathcal{P}$ be a family
of distributions satisfying Assumption \ref{assu:known ps}. Let $\mathcal{G}$
satisfy Assumption \ref{assu:known ps} (VC). Then the estimation
errors vanishes
\begin{align*}
\sup_{P\in\mathcal{P}}\sup_{g\in\mathcal{G}}\left|\frac{1}{\sqrt{n}}\sum_{i}(\tau_{i}^{\ast}-\widetilde{\tau}_{i})\cdot g(X_{i})\right| & \rightarrow_{p}0\\
\sup_{P\in\mathcal{P}}\sup_{g\in\mathcal{G}}\left|\frac{1}{\sqrt{n}}\sum_{i}(C_{i}^{\ast}-\widetilde{C}_{i})\cdot g(X_{i})\right| & \rightarrow_{p}0
\end{align*}
\end{lem}

\subsubsection{Proofs of auxiliary lemmas}

\begin{proof}
\textbf{Proof of Lemma \ref{lem:vc-subgraph}. }This lemma follows
directly from Lemma A.1 of \citet{kitagawa_who_2018}.
\end{proof}
\begin{proof}
\textbf{Proof of Lemma \ref{lem:identification}.} Under Assumption
\ref{assu:known ps}, we prove each $\widetilde{\tau}_{i}$ is an
conditionally unbiased estimate for $\tau$: 
\[
\E_{P}[\widetilde{\tau}_{i}g(X_{i})]=\E_{P}[\E_{P}[\widetilde{\tau}_{i}\mid X_{i}]g(X_{i})]=\E_{P}[\E_{P}[\tau\mid X_{i}]g(X_{i})]=W(g;P).
\]
We focus on $\E_{P}[\widetilde{\tau}_{i}\mid X_{i}]=\E_{P}[\E_{P}[\widetilde{\tau}_{i}\mid V_{i},X_{i}]\mid X_{i}]$.
Dropping the subscript $P$ for simplicity, conditional on $V_{i}$, by unconfoundedness and strict overlap
we have 
\begin{align*}
\E[\widetilde{\tau}_{i}\mid V_{i},X_{i}] & =\E[Y_{i}\mid V_{i},X_{i},D_{i}=1]-\E[Y_{i}\mid V_{i},X_{i},D_{i}=0]\\
 & =\E[Y_{i}(1)-Y_{i}(0)\mid V_{i},X_{i}]=\E[\tau_{i}\mid V_{i},X_{i}]
\end{align*}
Specifically 
\begin{align*}
\E[\widetilde{\tau}_{i}\mid V_{i},X_{i}] & =\E\left[\gamma^{Y}(V_{i},1)-\gamma^{Y}(V_{i},0)\mid V_{i},X_{i}\right]+\E\left[\alpha(V_{i},D_{i})\cdot\left(Y_{i}-\gamma^{Y}(V_{i},D_{i})\right)\mid V_{i},X_{i}\right]\\
 & =\E\left[Y_{i}\mid V_{i},D_{i}=1\right]-\E\left[Y_{i}\mid V_{i},D_{i}=0\right]+\\
 & \E\left[Y_{i}(1)-Y_{i}(0)\mid V_{i},X_{i}\right]-(\E\left[Y_{i}\mid V_{i},D_{i}=1\right]-\E\left[Y_{i}\mid V_{i},D_{i}=0\right])\\
 & =\E[\tau\mid V_{i},X_{i}]
\end{align*}
Specifically, we expand $\E[\alpha(V_{i},D_{i})Y_{i}\mid V_{i},X_{i}]$
\begin{align*}
 & =\E\left[\frac{Y_{i}}{p(V_{i})}\mid D_{i}=1,V_{i},X_{i}\right]\cdot \Pr\{D_{i}=1\mid V_{i},X_{i}\}-\E\left[\frac{Y_{i}}{1-p(V_{i})}\mid D_{i}=0,X_{i}\right]\cdot \Pr\{D_{i}=0\mid V_{i},X_{i}\}\\
 & =\E\left[D_{i}Y_{i}(1)+(1-D_{i})Y_{i}(0)\mid D_{i}=1,V_{i},X_{i}\right]-\E\left[D_{i}Y_{i}(1)+(1-D_{i})Y_{i}(0)\mid D_{i}=0,V_{i},X_{i}\right]\\
 & =\E\left[Y_{i}(1)\mid V_{i},X_{i}\right]-\E\left[Y_{i}(0)\mid V_{i},X_{i}\right]
\end{align*}
where the second line holds by independent characteristic such that
\[
\Pr\{D_{i}=1\mid V_{i},X_{i}\}=\Pr\{D_{i}=1\mid V_{i}\}\eqqcolon p(V_{i})
\]
and the last line follows from unconfoundedness. Similarly, we show
\[
\E\left[\alpha(V_{i},D_{i})\cdot\gamma^{Y}(V_{i},D_{i})\mid V_{i},X_{i}\right]=\E\left[Y_{i}\mid V_{i},D_{i}=1\right]-\E\left[Y_{i}\mid V_{i},D_{i}=0\right]
\]

Similar argument holds for $\E_{P}[\widetilde{C}_{i}\cdot g(X_{i})]$,
with the only modification:
\begin{align*}
 & \E[\frac{D_{i}}{p(V_{i})}H_{i}\mid V_{i},X_{i}]=\E\left[\frac{H_{i}}{p(V_{i})}\mid D_{i}=1,V_{i},X_{i}\right]\cdot \Pr\{D_{i}=1\mid V_{i},X_{i}\}\\
 & =\E[H_{i}(1)\mid V_{i},X_{i}]
\end{align*}
\end{proof}

\begin{proof}
\textbf{Proof of Lemma \ref{lem:known ps Glivenko=002013Cantelli}.}
Denote the following set of functions from $\mathcal{U}$ to $\mathbb{R}$
\begin{align*}
\mathcal{F}^{U} & =\{U_{i}\cdot g(X_{i}):g\in\mathcal{G}\}
\end{align*}
and it has uniform envelope $\bar{F}=M/2$ since $U_{i}$ is bounded.
This envelop function is bounded uniformly over $\mathcal{P}$. Also,
by Assumption \ref{assu:known ps} (VC) and Lemma \ref{lem:vc-subgraph},
$\mathcal{F}^{U}$ is VC-subgraph class of functions with VC-dimension
at most $v$. By Lemma 4.14 and Proposition 4.18 of \citet{wainwright_high-dimensional_2019},
we conclude that $\mathcal{F}^{U}$ has Rademacher complexity  $2\sqrt{M^{2}\frac{v}{n}}$.
Then by Proposition 4.12 of \citet{wainwright_high-dimensional_2019}
we conclude that $\mathcal{F}^{U}$ are $P$-Glivenko--Cantelli for
each $P\in\mathcal{P}$, with an $O(\sqrt{\frac{v}{n}})$ rate of
convergence. Note that this argument does not use any constants that
depend on $P$ but only $M$ and $v$, so we can actually get uniform
convergence over $\mathcal{P}$.
\end{proof}
\begin{proof}
\textbf{Proof of Lemma \ref{lem:known ps Donsker}.} Note that Assumption
\ref{assu:known ps}  imply that $\mathcal{F}^{W}$ and $\mathcal{F}^{B}$
have uniform envelope $\bar{F}=M/(2\kappa)$. $\mathcal{F}^{W}$ and
$\mathcal{F}^{B}$ thus have square integrable envelop functions uniformly
over $\mathcal{P}$. Also, by Assumption \ref{assu:known ps} (VC)
and Lemma \ref{lem:vc-subgraph}, $\mathcal{F}^{W}$ and $\mathcal{F}^{B}$
are VC-subgraph class of functions with VC-dimension at most $v$.
Even though both $\mathcal{F}^{W}$ and $\mathcal{F}^{B}$ depend
on $P$, a similar argument for Theorem 1 in \citet{rai_statistical_2019}
show that $\mathcal{F}^{W}$ and $\mathcal{F}^{B}$ are $P$-Donsker
uniformly in $P\in\mathcal{P}$. 
\end{proof}
\begin{proof}
\textbf{Proof of Lemma \ref{lem:ps vanish estimation error}. }We
focus on the deviation in $\tau_{i}^{\ast}-\widetilde{\tau}_{i}$.
The deviation in $C_{i}^{\ast}-\widetilde{C}_{i}$ can be proven to
vanish in a similar manner. Denote $\Delta\gamma^{Y}(V_{i})=\gamma^{Y}(V_{i},1)-\gamma^{Y}(V_{i},0)$.
For any fixed policy $g$, we expand the deviation $\frac{1}{\sqrt{n}}\sum_{i}(\tau_{i}^{\ast}-\widetilde{\tau}_{i})g(X_{i})$
into three terms
\begin{align}
 & \frac{1}{\sqrt{n}}\sum_{i}g(X_{i})\left(\widehat{\alpha}(V_{i},D_{i})-\alpha(V_{i},D_{i})\right)\cdot\left(Y_{i}-\gamma^{Y}(V_{i},D_{i})\right)\label{eq:estimation error}\\
 & +\frac{1}{\sqrt{n}}\sum_{i}g(X_{i})\left(\Delta\widehat{\gamma}^{Y}(V_{i})-\Delta\gamma^{Y}(V_{i})-\alpha(V_{i},D_{i})\cdot\left(\widehat{\gamma}^{Y}(V_{i},D_{i})-\gamma^{Y}(V_{i},D_{i})\right)\right)\nonumber \\
 & -\frac{1}{\sqrt{n}}\sum_{i}g(X_{i})\left(\widehat{\gamma}^{Y}(V_{i},D_{i})-\gamma^{Y}(V_{i},D_{i})\right)\cdot\left(\widehat{\alpha}(V_{i},D_{i})-\alpha(V_{i},D_{i})\right)\cdot\nonumber 
\end{align}
Denote these three summands by $D_{1}(g)$, $D_{2}(g)$ and $D_{3}(g)$.
We will bound all three summands separately. Recall we use the full
sample to estimate the propensity score and the CEF with a saturated
model. The purpose of the above expansion is to separately bound the
estimation error from the estimated CEF and propensity score, and
the deviation from taking sample averages. For cross-fitted estimators
for the propensity score and the CEF, a similar bound can be found
in \citet{athey_policy_2021}.

\paragraph{Uniform consistency of the estimated CEF and propensity score}

Denote with $b(V_{i},D_{i})$ the dictionary that spans $(V_{i},D_{i})$,
and $b(V_{i},)$ the dictionary that spans $V_{i}$. The saturated
models are therefore parameterized as $\gamma^{Z}(V_{i},D_{i})=\gamma'b(V_{i},D_{i})$
and $p(V_{i})=\beta'b(V_{i})$. Under standard argument, the OLS estimators
$\widehat{\gamma}$ and $\widehat{\beta}$ are asymptotically normal
uniformly over $P\in\mathcal{P}$:
\[
\sqrt{n}\cdot\left(\widehat{\gamma}-\gamma\right)=O_{P}(1),\ \sqrt{n}\cdot\left(\widehat{\beta}-\beta\right)=O_{P}(1).
\]
Furthermore, the in-sample $L_{2}$ errors from the estimated CEF
and propensity score vanish. Consider
\[
\frac{1}{n}\sum_{i}\left(\widehat{\gamma}^{Z}(V_{i},D_{i})-\gamma^{Z}(V_{i},D_{i})\right)^{2}=\frac{1}{n}\sum_{i}\left(\left(\widehat{\gamma}-\gamma\right)'b(V_{i},D_{i})\right)^{2}=\left(\widehat{\gamma}-\gamma\right)'\widehat{M}\left(\widehat{\gamma}-\gamma\right)
\]
where $\widehat{M}=\frac{1}{n}\sum_{i}b(V_{i},D_{i})b(V_{i},D_{i})'.$
It converges in probability to a fixed matrix $M=\E[b(V_{i},D_{i})b(V_{i},D_{i})']$.
So the in-sample $L_{2}$ error from the estimated CEF vanishes at
the rate of $n^{-1/2}$.

Similarly, consider expanding $\frac{1}{n}\sum_{i}\left(\widehat{\alpha}(V_{i},D_{i})-\alpha(V_{i},D_{i})\right)^{2}$
as 
\[
\frac{1}{n}\sum_{i}\left(\frac{1}{\widehat{\beta}'b(V_{i})}-\frac{1}{\beta'b(V_{i})}\right)^{2}D_{i}^{2}+\left(\frac{1}{1-\widehat{\beta}'b(V_{i})}-\frac{1}{1-\beta'b(V_{i})}\right)^{2}\left(1-D_{i}\right)^{2}
\]
With a first-order Taylor approximation, for each term in the summand,
the dominating term would be
\begin{align*}
\frac{1}{n}\sum_{i}\left(\left(\widehat{\beta}-\beta\right)'\frac{-b(V_{i})}{(\beta'b(V_{i}))^{2}}\right)^{2}D_{i}^{2} & =\left(\widehat{\beta}-\beta\right)'\left(\frac{1}{n}\sum_{i}\frac{b(V_{i})b(V_{i})'}{(\beta'b(V_{i}))^{2}}D_{i}^{2}\right)\left(\widehat{\beta}-\beta\right)
\end{align*}
where the middle term converges to a fixed matrix as implied by $\beta'b(V_{i})$
being bounded away from zero and one. So the in-sample $L_{2}$ error
from the estimated propensity score also vanishes at the rate of $n^{-1/2}$.

\paragraph{Bounding the deviation}

We now bound each term in (\ref{eq:estimation error}). Plugging in
the first-order Taylor approximation with a remainder term to the
estimated propensity score, we have 
\begin{align*}
D_{1}(g)= & \frac{1}{\sqrt{n}}\sum_{i}g(X_{i})\left(Y_{i}-\gamma^{Y}(V_{i},D_{i})\right)\cdot\\
 & \left(\left(\frac{1}{\widehat{\beta}'b(V_{i})}-\frac{1}{\beta'b(V_{i})}\right)D_{i}+\left(\frac{1}{1-\widehat{\beta}'b(V_{i})}-\frac{1}{1-\beta'b(V_{i})}\right)\left(1-D_{i}\right)\right)\\
= & \frac{1}{\sqrt{n}}\sum_{i}g(X_{i})\left(\widehat{\beta}-\beta\right)'\frac{-b(V_{i})}{(\beta'b(V_{i}))^{2}}\cdot D_{i}\cdot\left(Y_{i}-\gamma^{Y}(V_{i},D_{i})\right)+\\
 & \frac{1}{\sqrt{n}}\sum_{i}g(X_{i})\left(\widehat{\beta}-\beta\right)'\frac{b(V_{i})b(V_{i})'}{(\widetilde{\beta}'V_{i})^{3}}\left(\widehat{\beta}-\beta\right)\cdot D_{i}\cdot\left(Y_{i}-\gamma^{Y}(V_{i},D_{i})\right)\\
= & \underbrace{\sqrt{n}\left(\widehat{\beta}-\beta\right)}_{O_{P}(1)}'\underbrace{\frac{1}{n}\sum_{i}g(X_{i})\frac{-b(V_{i})}{(\beta'b(V_{i}))^{2}}\cdot D_{i}\cdot\left(Y_{i}-\gamma^{Y}(V_{i},D_{i})\right)}_{o_{P}(1)}+o_{P}(1)
\end{align*}
where $\widetilde{\beta}$ is a sequence between $\widehat{\beta}$
and $\beta$. This remainder term therefore converges to zero. Uniform
convergence of the sample average follows from Lemma \ref{lem:known ps Glivenko=002013Cantelli}:
the random vector $\frac{-b(V_{i})}{(\beta'b(V_{i}))^{2}}\cdot D_{i}\cdot\left(Y_{i}-\gamma^{Y}(V_{i},D_{i})\right)$
is mean-zero, has fixed dimension, and bounded.

We can decompose $D_{2}(g)$ into the product of two terms 
\[
D_{2}(g)=\underbrace{\sqrt{n}\left(\widehat{\gamma}-\gamma\right)}_{O_{p}(1)}'\frac{1}{n}\sum_{i}g(X_{i})\left(\Delta b(V_{i},D_{i})-\alpha(V_{i},D_{i})\cdot b(V_{i},D_{i})\right)
\]
Uniform convergence of the sample average again follows from Lemma
\ref{lem:known ps Glivenko=002013Cantelli}. We thus conclude $D_{1}(g)$
and $D_{2}(g)$ vanish uniformly over $g\in\mathcal{G}$ and over
$P\in\mathcal{P}$.

For $D_{3}(g)$, we apply the Cauchy-Schwarz inequality to note that
\[
D_{3}(g)\leq\sqrt{n}\cdot\sqrt{\frac{1}{n}\sum_{i}\left(\widehat{\gamma}^{Y}(V_{i},D_{i})-\gamma^{Y}(V_{i},D_{i})\right)^{2}}\cdot\sqrt{\frac{1}{n}\sum_{i}\left(\widehat{\alpha}(V_{i},D_{i})-\alpha(V_{i},D_{i})\right)^{2}}
\]
The terms in the square root are the in-sample $L_{2}$ errors from
the estimated CEF and propensity score, which vanish at the rate of
$n^{-1/2}$ uniformly over $P\in\mathcal{P}$ as shown in the paragraph
above. We thus conclude $D_{3}(g)$ vanishes uniformly over $g\in\mathcal{G}$
and over $P\in\mathcal{P}$.
\end{proof}

\section{Additional Theoretical Results and Discussion}\label{appx:addl theory}
\begin{prop}
\label{prop: EWM ptwise correct} Consider the same setting as in Proposition~\ref{prop: EWM no ptwise mistake}, except that  we have $\Pr_{P}\{C=1\mid X\}\in(0,1)$ almost surely such that the budget function strictly increases for $t\in[\underline{t},\overline{t}]$. Similarly, suppose  the budget constraint is at $k = \E_{P} [C \cdot \mathbf{1}\{X \leq \underline{t}\}].$  Then the sample-analog
rule $\widehat{g}_{\text{sample}}$  is asymptotically welfare efficient.
\end{prop}
\begin{proof}\textbf{Proof of Proposition~\ref{prop: EWM ptwise correct}.}
Since $\Pr_{P}\{C=1|X\}>0$ almost surely, we have $B(t;P)=\E_{P}[C\cdot\mathbf{1}\{X\leq t\}]=\Pr_{P}\{X\leq t\}\cdot \Pr_{P}\{C=1|X\leq t\}$ is strictly increasing in $t$.  The constrained optimal threshold is therefore the highest threshold where the constraint is satisfied exactly i.e. $t^{\ast}=\underline{t}.$

 Since $X$ is one-dimensional and continuously distributed, the welfare
function $W(t;P)$ is continuous in $t$. Since $W(t;P)$ is also
strictly increasing in $t$, for all $\epsilon>0$, there exists $t_{\epsilon}<t^{\ast}$
such that

\[
W(t_{\epsilon};P)=W(t^{\ast};P)-\epsilon.
\]

Since $C_{i}\cdot\mathbf{1}\{X_{i}\leq t\}$
is bounded for all $t$, the uniform law of large numbers applies and there exists a sequence $b_{n}\rightarrow0$
such that
\[
\lim_{n\rightarrow\infty}\Pr_{P^{n}}\left\{ \sup_{t}\left|\widehat{B}_{n}(t)-B(t;P)\right|\leq b_{n}\right\} =1.
\]
Since $t_{\epsilon}<t^{\ast}$ and
$B(t;P)$ continuous and strictly increasing, for $n$ large enough
(and $b_{n}$ small enough), we must have 

\[
B(t_{\epsilon};P)\leq B(t^{\ast};P)-b_{n}=k-b_{n}<k=B(t^{\ast};P).
\]
Taken together, we have
\[
\lim_{n\rightarrow\infty}\Pr_{P^{n}}\left\{ \widehat{B}_{n}(t_{\epsilon})\leq k-b_{n}+b_{n}\right\} \geq\lim_{n\rightarrow\infty}\Pr_{P^{n}}\left\{ \widehat{B}_{n}(t_{\epsilon})\leq B(t_{\epsilon};P)+b_{n}\right\} =1.
\]
As in the setting of Proposition~\ref{prop: EWM no ptwise mistake}, the policymakers know the welfare
function is strictly increasing, and the sample-analog rule  solves   $\max_{t}\widehat{B}_{n}(t)\text{ subject to }\widehat{B}_{n}(t)\leq k$.
However, the solution is not unique because $\widehat{B}_{n}(t)$
is a step function. For the proof, let the sample-analog rule
be the largest possible threshold to maximize $\widehat{B}_{n}(t)$:
\[
\widehat{t}=\max\left\{ \arg\max_{t}\{\widehat{B}_{n}(t)\text{ subject to }\widehat{B}_{n}(t)\leq k\}\right\} .
\] This implies that
$\lim_{n\rightarrow\infty}\Pr_{P^{n}}\left\{ t_{\epsilon}\leq\hat{t}\right\} =1$ and since $W(t;P)$ is strictly increasing in $t$, we have
\[
\lim_{n\rightarrow\infty}\Pr_{P^{n}}\left\{ W(\hat{t};P)\geq W(t^{\ast};P)-\epsilon\right\} =\lim_{n\rightarrow\infty}\Pr_{P^{n}}\left\{ W(\hat{t};P)\geq W(t_{\epsilon};P)\right\} =1
\]
which proves the asymptotic welfare efficiency.
\end{proof}
\section{Additional Empirical Results and Discussion}
\subsection{Robustness checks}\label{sec:robustness}
In this section, I present additional empirical results as I vary the significance level $\alpha$ and the penalty level $\bar\lambda$
 , relative to the benchmark values considered in Section~\ref{subsec:Estimating-a-more}.

I explore how the modified rule $\widehat{g}_{\alpha}$ changes as $\alpha$ varies from 5\% (as in the main text) to 50\% in Figure~\ref{fig:alpha 
 frontier} below. Each point on the curve represents the policy chosen by $\widehat{g}_{\alpha}$ for a corresponding statistical significance level $\alpha$. The vertical axis shows the  estimated welfare   and the horizontal axis shows the estimated budget of the selected policy. As $\alpha$ increases, the statistical test becomes less conservative and the upper confidence band becomes wider.  So the feasible region expands and the selected policy achieves higher estimated  welfare.
\begin{figure}[h]
    \centering
    \includegraphics[width=1\linewidth]{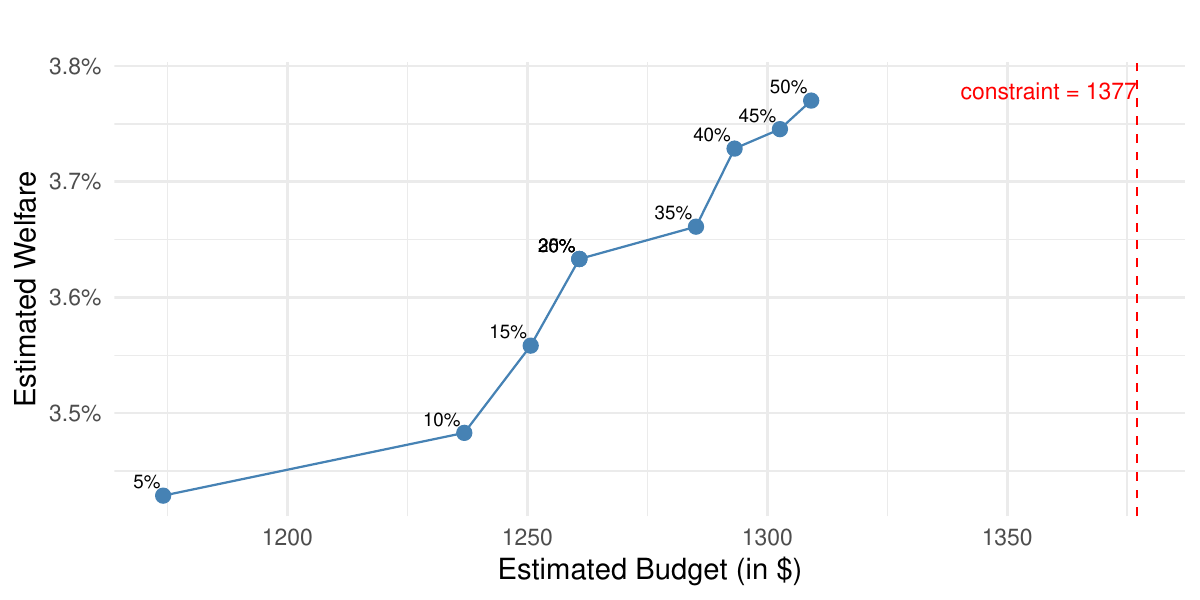}
    \caption{Frontier of $\widehat{g}_{\alpha}$ across statistical significance levels $\alpha$}
    \label{fig:alpha frontier}
\end{figure}

I also explore how $\widehat{g}_{\text{tradeoff}}$ changes for a range of values for $\overline{\lambda}$ as illustrated in the Figure~\ref{fig:frontier} below. Each point on the curve represents the policy chosen by $\widehat{g}_{\text{tradeoff}}$ for a corresponding  $\overline{\lambda}$. The vertical axis shows the  estimated welfare and the horizontal axis shows the estimated budget achieved by $\widehat{g}_{\text{tradeoff}}$. For example, if each dollar of budget overrun required a repayment of 1.68 times the amount, rather than just 1 as in the main text,  $\widehat{g}_{\text{tradeoff}}$ selects a policy that entails no estimated budget violation, aligning with $\widehat{g}_{\text{sample}}$.  Because the optimization is discrete in the sample, intermediate values of $\overline{\lambda}$ often result in the same selected policy, leading to gaps in the curve.
  
\begin{figure}[h]
    \centering
    \includegraphics[width=1\linewidth]{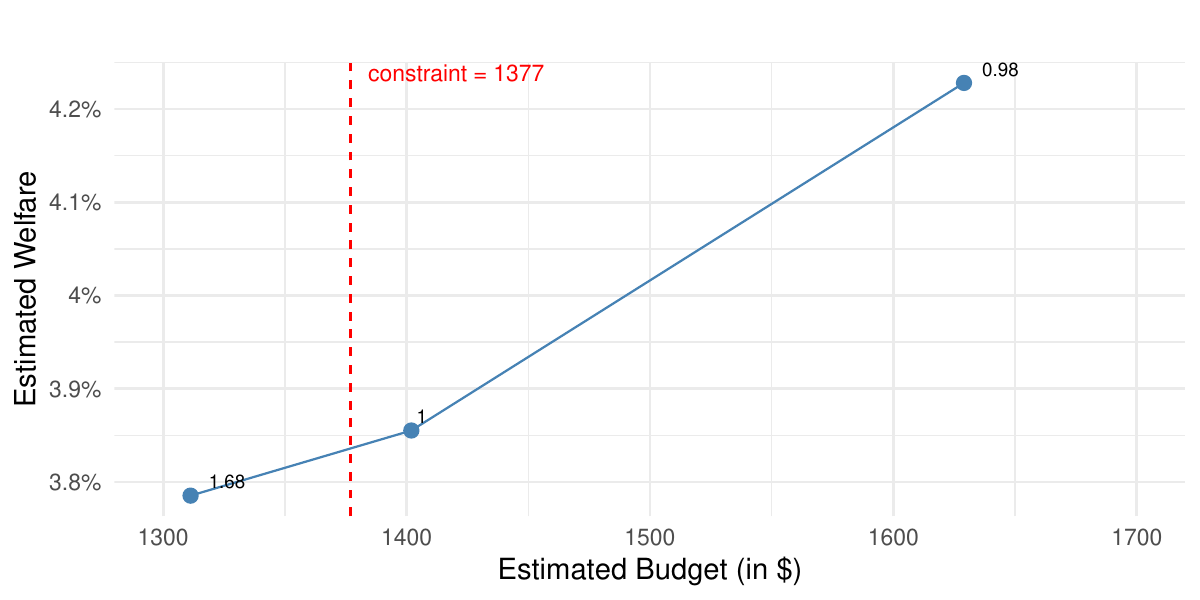}
    \caption{Frontier of $\widehat{g}_{\text{tradeoff}}$ across different penalty levels   $\overline{\lambda}$}
    \label{fig:frontier}
\end{figure}
\subsection{Computation}\label{appx:computation}
The policy class considered in this paper $g(x)=\mathbf1\{\beta'x\geq 0\}$ consists of linear eligibility score.  Therefore I follow \cite[Online Appendix C]{kitagawa_who_2018} to set up the problems of sample-analog rule (\ref{eq:EWM}) and trade-off rule (\ref{eq:hinge}) as a Mixed Integer Linear Programming (MILP), which is more efficient than solving them through a search over $x$. 
Due to the estimated constraint, I add  linear constraints to the version of \cite{kitagawa_who_2018}  and instead solve for the trade-off rule in two parts:
\begin{align}
\text{constrained} &  & \text{unconstrained}\label{eq:MILP discrete rewrite-bind}\\
\min_{\beta\in\mathbf{B},z\in\mathbb{R}^{n}}-\sum_{i=1}^{n}\tau_i^\ast z_i &  & \min_{\beta\in\mathbf{B},z\in\mathbb{R}^{n}}-\sum_{i=1}^{n}\tau_i^\ast z_i+\overline\lambda (\sum_{i=1}^{n}C_{i}^\ast z_i-n\cdot k)\nonumber \\
\text{ subject to }\frac{X_{i}'\beta}{\bar{C}_{i}}\leq z_{i}\leq1+\frac{X_{i}'\beta}{\bar{C}_{i}}, &  &\text{ subject to }\frac{X_{i}'\beta}{\bar{C}_{i}}\leq z_{i}\leq1+\frac{X_{i}'\beta}{\bar{C}_{i}},\nonumber \\
z_{i}\in\{0,1\}, &  & z_{i}\in\{0,1\},\nonumber \\
\sum_{i=1}^{n}C_{i}^\ast z_i <n\cdot k. &  & \sum_{i=1}^{n}C_{i}^\ast z_i \geq n\cdot k.\nonumber 
\end{align}
where constants $\bar{C}_{i}= \sup_{\beta\in B} |X_{i}'\beta|$. The additional binary parameters $(z_1,\dots,z_n)$  replace the policy functions $g(X_i)=\mathbf1\{\beta'X_i\geq 0\}$.  The sample-analog rule $\widehat{g}_{\text{sample}}$ is the constrained solution on the left hand side. The trade-off rule $\widehat{g}_{\text{tradeoff}}$ is the better solution across constrained and unconstrained policies in (\ref{eq:MILP discrete rewrite-bind}).

To modify the sample-analog rule to control the probability of selecting infeasible policies by $\gmistake$ as proposed in Section \ref{sec:Controlling-the-Mistake}, I introduce quadratic constraints to (\ref{eq:MILP discrete rewrite-bind}).
First, for a conventional level
of $\alpha$, constructing the tightened constraint $\hat{\mathcal{G}}_{\alpha}$ requires
an estimate for the critical value $c_{\alpha}$, the $\alpha\%$-quantile
from $\underset{g\in\mathcal{G}}{\inf}\ \frac{G_{P}^{B}(g)}{\Sigma_{P}^{B}(g,g)^{1/2}}$,
the infimum of the Gaussian process $G_{P}^{B}(\cdot)\sim\mathcal{GP}(0,\Sigma_{P}^{B}(\cdot,\cdot))$.
In practice, I construct a grid on $\mathcal{G}$ as
\begin{equation}
\mathcal{\widetilde{G}}=\left\{ g(x)=\mathbf{1}\{\text{income}\cdot1\{\text{numchild}=j\}\leq y_{j}\}:j\in\{0,1,\geq2\},y_{j}\in\{0,50,100,\dots,500\}\right\} \label{eq:coarse grid}
\end{equation}
for characteristics $x=(\text{income},\ \text{numchild})$. This grid thus consists
of income thresholds every 50\% of the federal poverty level, and
the thresholds can vary with number of children. I then approximate
the infimum over infinite-dimensional $\mathcal{G}$ by the minimum
over $\mathcal{\widetilde{G}}$ with estimated covariance i.e. $\underset{g\in\widetilde{\mathcal{G}}}{\min}\ \frac{\widetilde{h}(g)}{\widehat{\Sigma}^{B}(g,g)^{1/2}}$.
Here $\widetilde{h}(\cdot)\sim\mathcal{\mathcal{N}}(0,\widehat{\Sigma}^{B})$
is a Gaussian vector indexed by $g\in\widetilde{\mathcal{G}}$, with
$\widehat{\Sigma}^{B}$ is sub-matrix of the covariance estimate $\widehat{\Sigma}^{B}(\cdot,\cdot)$
for $g\in\widetilde{\mathcal{G}}$. Based on 10,000 simulation draws
I estimate $c_{\alpha}$ to be -2.47 for $\alpha=5\%$. The validity of this approximation
is given by the uniform consistency of the covariance estimator under
Assumption \ref{cond:known ps Donsker}.

Second, to select policies in $\hat{\mathcal{G}}_{\alpha}$ is equivalent to invert a test: $\frac{\sqrt{n}\left(\widehat{B}_{n}(g)-k\right)}{\widehat{\Sigma}^{B}(g,g)^{1/2}}\leq c_{\alpha}$. This test inversion can be written as two constraint (note that $c_{\alpha}<0$). The first constraint
is linear: $\sum_{i=1}^{n}C_i^\ast z_i \leq n\cdot k$ and the second is $n\cdot\left(\widehat{B}_{n}(g)-k\right)^{2}>c_{\alpha}^{2}\cdot\widehat{\Sigma}^{B}(g,g)$
which translates to a quadratic constraint
\begin{align*}
     &  n\left(\frac{\sum_{i=1}^{n}C_i^\ast z_i}{n}-k\right)^{2}>c_{\alpha}^{2}\left(\left(\frac{1}{n}\sum_{i=1}^{n}(C_i^\ast)^{2} z_i\right)-\left(\frac{\sum_{i=1}^{n}C_i^\ast z_i}{n}\right)^{2}\right)\\
   \Rightarrow  & \frac{c_{\alpha}^{2}}{n}\sum_{i=1}^{n}(C_i^\ast)^{2} z_i+2k\sum_{i=1}^{n}C_i^\ast z_i-\left(\frac{1}{n}+\frac{c_{\alpha}^{2}}{n^{2}}\right)(\sum_{i=1}^{n}C_i^\ast z_i)^2<nk^{2}
\end{align*}
 
\end{document}